\documentclass[final,3p,times]{elsarticle}
\usepackage{subcaption}
\usepackage{amsthm}
\usepackage{amsfonts}
\usepackage{mathtools}
\newtheorem{theorem}{Theorem}
\newtheorem{lemma}[theorem]{Lemma}

\usepackage{array}
\usepackage{float}
\usepackage{multirow}
\usepackage{hhline}
\usepackage{color}
\usepackage{booktabs}
	
\usepackage{amssymb}

\usepackage{amsmath}

\journal{Performance Evaluation}
\biboptions{sort&compress}

\begin{document}
	
	\begin{frontmatter}
		
		\title{Energy-performance tradeoffs in server farms with batch services and setup times}
		
		\author[label1]{Thu Le-Anh}
		\affiliation[label1]{organization={Graduate School of Science and Technology, University of Tsukuba},
			addressline={1-1-1 Tennodai},
			city={Tsukuba},
			postcode={305-8573},
			country={Ibaraki, Japan}}
		
		\author[label2]{Tuan Phung-Duc\corref{cor1}}
		\cortext[cor1]{Corresponding author.}
		\ead{tuan@sk.tsukuba.ac.jp}
		\affiliation[label2]{organization={Institute of Systems and Information Engineering, University of Tsukuba},
			addressline={1-1-1 Tennodai},
			city={Tsukuba},
			postcode={305-8573},
			country={Ibaraki, Japan}}
		
		\begin{abstract}
			
			Data centers consume a large amount of energy, much of which is wasted due to idle servers. Turning off idle servers might be an effective power-saving solution; however, there is a trade-off between energy savings and system performance. Hence, we propose a setup queueing model with a batching policy that allows servers to process a set of jobs simultaneously to minimize power consumption while maintaining acceptable performance. We consider an M/M/$c/$SET--BATCH queue, a multi-server batch service queue with a fixed batch size and setup times, and some variants, including systems in which idle servers delay before turning off or systems in which the batch size is dynamic. We analyze the steady-state probabilities and system performance of the M/M/$c/$SET--BATCH system and its variants. Our analysis of the M/M/$c/$SET--BATCH system with lower computational complexity is made possible by utilizing the special structure of the model. In addition, we use simulations to compare the M/M/$c/$SET--BATCH model with some other variants with different setup time distributions. The results suggest that the model performs better when the setup time has a larger coefficient of variation. Our results indicate that the batching policy enhances the system performance, especially when we allow servers to be idle before turning them off. 
		\end{abstract}
		
		\begin{keyword}
			Queueing analysis \sep Setup times \sep Batch services \sep Delay off \sep Energy consumption
			
		\end{keyword}
		
	\end{frontmatter}
	
	\section{Introduction}
	Data centers are major energy consumers and waste a substantial portion. In data centers, servers are busy only 10\%--30\% of the time; when idle, a server consumes around 60\% of its peak power usage \cite{barroso2007case}. Thus, turning idle servers off is an effective solution to save power. However, off servers need a considerable setup time to restart, during which they consume energy without processing jobs. Consequently, turning off idle servers makes a trade-off between energy conservation and system performance.
	
	In data centers, the setup time is typically 200 s, whereas the typical service time is less than 1 s \cite{gandhi2013exact}. Hence, many authors have proposed methods to minimize the number of times servers need to enter setup mode. The three main lines of research involve policies to delay turning off an idle server, called delayoff, policies to reduce setup times by using low-power sleep modes, and policies to restrict the number of servers that can be in setup mode simultaneously, called \textit{s}-staggered setup. Similar to the above research lines, this paper aims to minimize the number of times servers undergo setup by proposing a group processing policy. The group processing or batching policy allows a server to process a set of jobs simultaneously instead of a single job.
	
	Batch processing, particularly micro-batch processing, plays a vital role in data centers by optimizing resource utilization. Traditional batch processing systems such as Hadoop and Spark process large datasets in predefined batches over time, while modern micro-batching with tools such as Apache Spark, Apache Flink, Fluentd, and Logstash enables frequent small-group task execution \cite{carbone2015apache}. Hence, micro-batching processing supports near real-time processing compared with stream processing in real-time. Batch processing enables parallel task processes and supports job scheduling and ETL (extract, transform, and load) pipelines in cloud platforms like Amazon Web Services (AWS) and Google Cloud Batch \cite{carbone2015apache, wen2019energy}. A common situation in which the batching policy can yield benefits is when jobs are divided into families with similar characteristics (a family scheduling model). The main reasons for batch processing are to avoid server setups and optimize server utilization. Family scheduling or workflow scheduling with batch processing is typically enabled in cloud computing environments and data center networks that exist in E-commerce \cite{wen2019energy}. In addition, batch processing occurs frequently in transportation platforms, testing centers, or manufacturing companies where a number of items are processed simultaneously in batches. A typical example of batch processing can be observed in ride-sharing platforms such as Uber and Lyft, which are employed to optimize ride matching and route scheduling. By aggregating multiple ride requests within a short time (micro-batching), the systems determine optimal driver-passenger pairings, thereby reducing operational costs while maintaining service rates.
	
	Despite the critical applications of both setup times and batch services in various systems such as data centers, transportation platforms and manufacturing systems, most papers separately analyze the effect of setup times \cite{gandhi2010server, phung2017exact, maccio2018asymptotic, maccio2018structural} and batch services \cite{claeys2011analysis, baetens2019analysis, nain2021one, dudin2024analysis} in queueing systems. Some related models combine batch services and vacations without setup \cite{sikdar2008batch,gupta2020finite, saxena2018modeling}. It should be noted that the vacation schedule is independent, while setup rates depend on the queue length, making the model with setup times more complex. In \cite{arumuganathan2005steady}, a setup is taken into account in a single-server queueing model with batch services and vacations where a server takes multiple vacations of random duration when the queue length falls below a threshold. Upon returning from a vacation, the server requires a setup time before starting service if the queue length meets or exceeds the threshold. In contrast, our model allows the setup process to start immediately once the queue length reaches the threshold, regardless of vacation. 
	
	The most closely related work is \cite{le2023performance}, where the authors analyze a finite-buffer queueing model with setup times and batch services. They provide a simple analysis of the finite-buffer queue, which is limited to small-scale systems due to computational complexity. In practice, almost all job queues have a certain capacity, even in data centers with many servers. The queues are also typically constrained by a reasonable buffer size to ensure quality of service (QoS) regarding response times. Thus,  it is essential to compute performance measures for queueing systems with finite buffers. However, as the buffer capacity increases, calculating the joint stationary distribution becomes laborious due to the growing number of boundary equations. Although buffers are finite, they are often large, making it reasonable to approximate a finite-buffer system by the corresponding infinite-buffer system. Hence, we propose to analyze an infinite-buffer M/M/$c$/Setup queue with the batching policy and use its results to approximate the corresponding system with finite buffers. We present an evaluation of the efficiency of our approximation compared with the exact one for the finite-buffer queue in terms of time and accuracy.
	
	In addition, although many theoretical works have investigated the effect of setup times in multi-server systems, most studies have assumed that setup times are exponentially distributed. In reality, setup times in many systems are less variable than this, and deterministic setup times are common \cite{williams2022m}. For example, the setup time for servers in the data center in \cite{gandhi2012autoscale} was constant. The simulation results in \cite{le2023performance} showed that the negative effect of setup times on delay in a model with deterministic setup was greater than in an exponential setup time model. Accordingly, the assumption that setup times are i.i.d.\ exponentially distributed may lead to an underestimation of the effect of setup time on system performance. As a result, we first analyze the system performance under a Markovian setting; we then use simulations to study three other cases in which the setup time follows a deterministic, hyperexponential, and  Erlang distribution. 
	
	Finally, we consider some variants of our initial model with a delayoff policy that allows an idle server to delay before shutdown and a dynamic batch-size policy that allows the batch size to vary depending on the load. For the delayoff policy, although allowing idle times before shutdown can reduce the number of servers in setup mode, it consumes energy for idle servers. For the dynamic batch-size policy, the batch size can vary within a specific range depending on load instead of being fixed as in our original model. In this case, a server processes up to a certain number of jobs simultaneously, satisfying the upper and lower bounds of the server size (batch size). Intuitively, the dynamic batch-size policy is flexible, but studies have yet to analyze its effectiveness in a setup queueing model. Accordingly, analyzing and comparing the above models allows us to find optimal policies when system parameters change. Our main contributions are summarized below:
	
	1. A setup queueing model with a batching policy and several variant models are proposed to balance the trade-off between system performance and power consumption in data centers.
	
	2. An exact recursive algorithm is presented to compute effectively the joint stationary distribution and performance measures based on the special structure of the rate matrices in the setup queueing model with the batching policy.
	
	3. The efficiency of our method is evaluated by comparing its computational complexity with that of the direct matrix inversion method for computing the rate matrices, as well as by assessing the accuracy and computational complexity of our approximation for the finite-buffer case.
	
	4. The optimal policies corresponding to different system parameters are chosen based on numerical experiments and simulations.
	
	The rest of this paper is organized as follows. Section \ref{sec2} discusses some related works and Section \ref{sec3} presents a formal definition of the models.  In Section \ref{sec4} and \ref{sec5}, we summarize the main results including the analysis of the steady-state probabilities and the system performance, respectively. Section~\ref{sec6} is devoted to the numerical experiments and simulations, and concluding remarks are presented in Section \ref{sec7}.
	
	\section{Related works}\label{sec2}
	In this section, we discuss papers that analyze the effect of the setup time and the batch service on the performance of multi-server queueing systems. Although we have found no analysis of multi-server queues with batch services and setup times, queueing models with setups and queueing models with batch services have been extensively investigated.
	
	\subsection{Batch service without setup}
	Bailey \cite{bailey1954queueing} was the first to study a single-server queueing system with a batch service. This work analyzed an M/G/1 queue in which the services are provided in batches of varying sizes. Since then, many papers have been published about batch services; see, for example, \cite{claeys2011analysis, baetens2019analysis, nain2021one, dudin2024analysis, adan2000multi, gai2024waiting}.
	In most of batch service queues, customers are served in batches of sizes varying within a certain range, and this service rule is referred to by Neuts \cite{neuts1967general} as the ``general bulk service'' (GBS) rule.  Several studies have proposed combining batch services with the GBS rule and vacations, where a server enters a vacation if the queue length is less than a threshold corresponding to the minimum batch size. The server resumes service upon the completion of the vacation, provided the queue length exceeds this threshold \cite{sikdar2008batch,gupta2020finite, saxena2018modeling}.
	
	\subsection{M/M/$c$ with setup}
	Several studies have considered the optimal policies of multi-server setup queues for energy-efficient management of data centers \cite{gandhi2010server, phung2017exact, maccio2018asymptotic, maccio2018structural}. Gandhi et al.\ \cite{gandhi2010server} presented various approximations for setup queues in which a server is turned off immediately when idle. These approximations work well when load or setup times or both are low. Gandhi et al. \cite{gandhi2013exact} developed a new analysis technique, called recursive renewal reward (RRR) approach to give the first exact, closed-form analysis for the M/M/$c$/Setup queue. Phung-Duc \cite{phung2017exact} comprehensively compared the matrix analytic approach and the generating function method and derived the exact solutions as in \cite{gandhi2013exact}.
	
	\subsection{\textit{s}-staggered setup}
	The staggered setup model allows only one server to be in setup mode at a time, simplifying the analysis. In \cite{artalejo2005analysis}, Artalejo et al.\ analyzed a multi-server queue with a staggered setup policy and derived the stationary queue length distribution using difference equations. The model in \cite{artalejo2005analysis} was extended to take into account batch arrivals in~\cite{PhungDuc2014Staggered}.
	
	Many papers have considered the \textit{s}-staggered setup policy \cite{gandhi2010server,phung2016energy}, which extends the staggered policy by limiting the number of setup servers at a time to \textit{s}. In \cite{gandhi2010server}, Gandhi et al.\ presented the first analysis of the response time distribution and decomposition results for the mean response time. In \cite{phung2016energy}, Phung-Duc and Kawanishi presented a simple recursion extending \cite{PhungDuc2015MMcK} to obtain the steady-state stationary distribution of a multi-server queue with $s$-staggered policy and abandonment.
	
	\subsection{M/M/$c$/Setup/Delayoff}
	The delayoff policy stems from the simple idea that when a server is idle, the server will wait for some time instead of shutting down immediately. Recently, there have been a small number of studies analyzing a delayoff policy in M/M/$c$/Setup queueing models. Gandhi et al.\ \cite{gandhi2013exact} considered an M/M/$c$/Setup model with the delayoff policy based on the RRR approach and developed an approximation for choosing the amount of time before turning an idle server off. In \cite{pender2016law}, Pender and Phung-Duc analyzed an M/M/$c$/Setup--Delay model incorporating a stochastic number of servers, a non-stationary arrival rate, and job abandonment. The authors obtained a mean field approximation for the mean queue length and the mean number of non-idle servers and showed that they converge as the number of servers and the arrival rate approach infinity.
	
	\subsection{M/M/$c$/Setup/Sleep}
	The M/M/$c$/Setup model with sleep mode is motivated by the idea that idle servers can be put to sleep or turned off. Although a server that is turned off does not consume power, it requires a considerable setup time. By contrast, a server in sleep mode consumes less power than the idle power while requiring a setup time typically smaller than the setup time of a powered-down server \cite{meisner2009powernap, gandhi2012sleep}. Therefore, there exists a trade-off between putting an idle server to sleep and shutting it down. In 2009, Eisner \cite{meisner2009powernap} presented a new energy-saving method called Powernap, designed to quickly switch between a high-performance active state and a low-power nap state in response to load. Their experimental results based on utilization traces collected from real data centers showed that PowerNap reduces energy consumption by 74\% on average. Gandhi et al.\ \cite{gandhi2013exact} used the RRR method to give an explicit closed-form solution for the M/M/$c$/Setup/Sleep queue. It can be used to find exact solutions for many Markov chains beyond M/M/$c$/Sleep such as M/E\textsubscript{2}/$c$, in which service times follow the Erlang-2 distribution, and M\textsuperscript{2}/M/$c$, in which jobs come in batches. 
	\subsection{Deterministic setup times}
	Despite the importance of deterministic setup times, there has been little research on the M/M/$c$/Setup--deterministic queue. Recently, William et al. \cite{williams2022m} presented the first analysis of this queue and derived a lower bound for the mean waiting time. The bound is highly accurate when the setup rate is much lower than the service rate. The M/M/$c$/Setup queue with deterministic setup times has also been analyzed via simulations in \cite{le2023performance,kara2017energy}. In \cite{le2023performance}, the authors considered a finite-buffer multi-server setup queue with a batching policy in which the setup time is deterministic or follows an Erlang, exponential, or hyperexponential distribution. Their simulation results show that the performance of the model is better when the setup time has a larger coefficient of variation. These studies agree that the adverse effect of setup times on the mean delay in the M/M/$c$/Setup--deterministic model is much greater than in the corresponding model with exponential setup times. Other than the above, progress has yet to be made in multi-server queueing systems with deterministic setup times.
	
	\section{Model}\label{sec3}
	
	In our models, there are $c$ servers, jobs arrive according to a Poisson process with rate $\lambda$, and they are grouped and processed in batches. The service time of batches follows an exponential distribution with mean $1/\mu$. Since a server requires some random setup time to transition from ``off" to ``on," we assume that the setup times are exponentially distributed with mean $1/\alpha$. 
	\subsection{ M/M/$c$/SET--BATCH}
	We now formally describe the M/M/$c$/SET--BATCH model, in which SET and BATCH stand for ``setup'' and ``batching policy'', respectively. In this model, a server will turn off if no batches request service upon service completion. Otherwise, it immediately takes the waiting batch at the head of the queue to process. At that time, the server with the largest remaining setup time is turned off if the number of batches left in the queue is less than the number of setup servers, that is, we are setting up more servers than needed. 
	
	Let $i$, $j$, and $b$ denote the number of active servers, the number of jobs in the system, and the batch size, respectively. Then, the number of setup servers is given by min$(\lfloor\frac{j-ib}{b}\rfloor,c-i)$. The number of active servers cannot exceed the number of batches in the system and thus each server has three states: ``setup,'' ``busy,'' and ``off.''  We assume that servers in this system operate in a first come, first served (FCFS) manner. Fig.~\ref{Fig1} shows the state transition diagram for the case $b=2$. 
	
	For the finite-buffer case, the M/M/$c/K$/SET--BATCH model, the setting is identical except that the number of jobs in the system is limited to a maximum of $K$, for $K \geq cb$.
	\begin{figure}
		\centering   \includegraphics[width=0.8\textwidth]{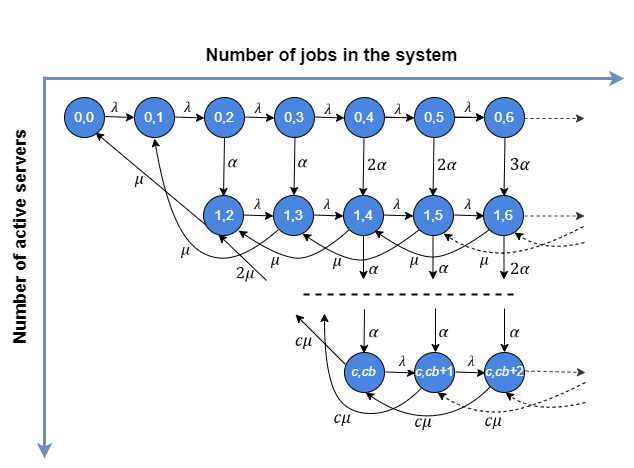}
		\caption{State transition diagram for M/M/$c$/SET--BATCH queue ($b=2$).}
		\label{Fig1}
	\end{figure}
	\subsection{ M/M/$c$/SET--BATCH/Delayoff}
	M/M/$c$/SET--BATCH/Delayoff is the same as M/M/$c$/SET--BATCH, except that idle servers are not immediately turned off. We assume that when a server completes service and there are no batches waiting in the queue, the server will remain in the idle state for an exponentially distributed amount of time with mean $1/\beta$. During that time, if no batches request service, the idle server will turn off. Otherwise, the idle server instantly takes the incoming batch to process.
	
	Let $i$, $j$, and $b$ denote the number of busy or idle servers,
	the number of jobs in the system, and the batch size, respectively.
	Fig.~\ref{Fig2} shows the state transitions for the
	M/M/$c$/SET--BATCH/Delayoff system in the case $b=2$.
	The number of setup servers is given by
	$\max(\min(\lfloor\frac{j-ib}{b}\rfloor,c-i),0)$
	and the number of idle servers is
	$\max(i-\lfloor\frac{j}{b}\rfloor,0)$.
	Specifically, if $j\geq(i+1)b$, the number of setup servers is
	$\min(\lfloor\frac{j-ib}{b}\rfloor,c-i)$
	and there are no idle servers.
	If $j<ib$ (gray states), the number of idle servers is
	$i-\lfloor\frac{j}{b}\rfloor$ and there are no setup servers.
	If $ib\leq j<(i+1)b$, there are no idle or setup servers.
	
	\begin{figure}
		\centering   \includegraphics[width=0.8\textwidth]{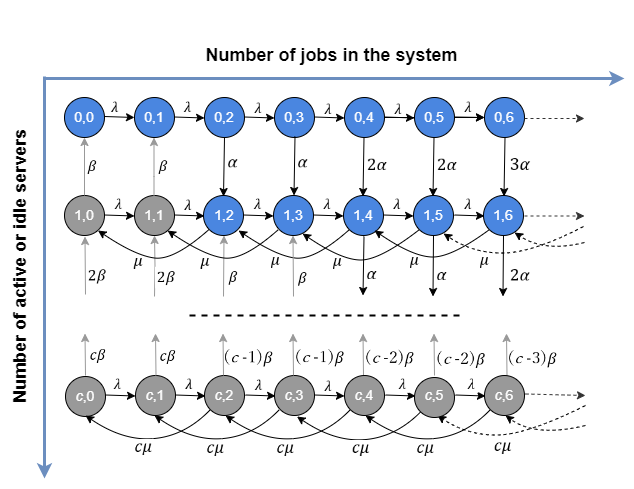}
		\caption{State transition diagram for M/M/$c$/SET--BATCH/Delayoff queue ($b=2$).}
		\label{Fig2}
	\end{figure}
	\subsection{ M/M/$c$/SET--BATCH/Dynamic}
	The M/M/$c$/SET--BATCH/Dynamic model is similar to the M/M/$c$/SET--BATCH model, except that the batch size is not fixed but varies within a certain range ($b_1$ to $b_2$). In particular, a server will be set up if enough $b_1$ customers are waiting in the queue. It also assumes an exponentially distributed setup time, as in the M/M/$c$/SET--BATCH model. However, upon the completion of the setup process, the server can serve as many customers as possible provided that it does not exceed $b_2$. After the service is completed, the server will turn off if there are fewer than $b_1$ customers in the queue. Otherwise, it will immediately serve as many customers as possible but no larger than $b_2$. Therefore, the batch size can vary between $b_1$ and $b_2$ in the M/M/$c$/SET--BATCH/Dynamic model.
	
	Fig.~\ref{Fig3} shows the state transitions for the M/M/$c$/SET--BATCH/Dynamic model in the case $b_1=2$ and $b_2=3$. Let $i$ and $j$ denote the number of active servers and the number of jobs in the queue, respectively. The number of setup servers is given by min$(\lfloor\frac{j}{b_2}\rfloor + \text{min}(\lfloor\frac{j \mod b_2}{b_1}\rfloor, 1),c-i)$.
	\begin{figure}
		\centering   \includegraphics[width=0.75\textwidth]{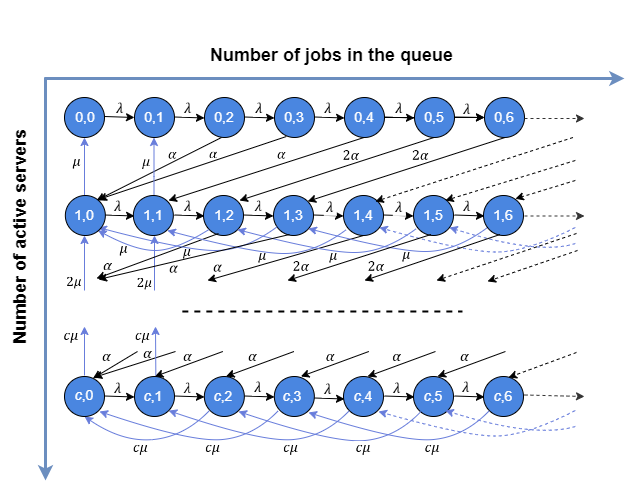}
		\caption{State transition diagram for M/M/$c$/SET--BATCH/Dynamic queue ($b_1=2,$ $b_2=3$).}
		\label{Fig3}
	\end{figure}
	
	\section{The quasi-skip-free process (QSF)}\label{sec4}
	In this section, we calculate the joint stationary distribution of the Markov processes to obtain the corresponding stationary performance measures of the three models defined in Section \ref{sec3}. For the M/M/$c$/SET--BATCH model, we derive a recursive algorithm to calculate the stationary distribution and performance measures using a matrix analytic approach. In this model, all non-diagonal elements of the rate matrix $R$ in homogeneous part are explicitly expressed in terms of the diagonal elements, which are the unique root of a polynomial (see Eq.~(\ref{eq:2})), instead of solving the matrix equation as in general matrix analytic methods. Furthermore, the boundary part is solved using rate matrices which are efficiently obtained as in Section \ref{sec:4.1.2}. For the M/M/$c$/SET--BATCH/Delayoff and M/M/$c$/SET--BATCH/Dynamic models, we use a standard matrix geometric method as in Section \ref{sec:4.2} and Section \ref{sec:4.3}. All models above can be formulated as a quasi-skip-free (QSF) process to the right. The QSF Markov chains capture the dynamics of batch service queues where a server can process multiple jobs simultaneously. The jump rates in the QSF process are not allowed to cross two or more levels away from the current state in one direction of the level; that is, the right in the case of our models 
	\cite{neuts1994matrix}. 
	\subsection{M/M/$c$/SET--BATCH}
	In the M/M/$c$/SET--BATCH model, all the servers will eventually be active if there are enough jobs in the system, so the system stability condition is $\lambda < cb\mu$. Let $B(t)$ and $N(t)$ denote the number of busy servers and the number of jobs in the system at time $t$, respectively. In the current setting, the two-dimensional process $\{X(t)=(B(t), N(t));t\geq0\}$ forms a Markov chain in the state space 
	\begin{equation*}
		S=\{(i,j);i=0,1,\ldots,c,j=ib,ib+1,\ldots\}.
	\end{equation*}
	The system can be formulated as a quasi-skip-free process (QSF) with the infinitesimal generator $Q$ given as
	
	\begin{equation*}
		Q=\left(\begin{array}{cccccccccccc}
			Q_{0}^{(0)} & Q_{1}^{(0)} & \mathit{O} & \mathit{O} & \mathit{O} & \mathit{O} & \cdots \\
			\mathit{O} & Q_{0}^{(1)} & Q_{1}^{(1)} & \mathit{O} & \mathit{O} & \mathit{O} & \cdots \\
			\vdots & \ddots & \ddots & \ddots & \ddots & \ddots & \cdots \\
			Q_{-b}^{(b)} & \mathit{O} & \cdots &  Q_{0}^{(b)} & Q_{1}^{(b)}  & \ddots & \cdots \\
			\mathit{O} & Q_{-b}^{(b+1)} & \mathit{O} & \cdots & Q_{0}^{(b+1)} & Q_{1}^{(b+1)} & \cdots \\
			\vdots & \vdots & \vdots & \vdots & \vdots & \vdots & \ddots 
		\end{array}\right) \text { , }
	\end{equation*}
	where $\mathit{O}$ denotes the zero matrix with an appropriate dimension. We use the notation $Q_{1}^{(i)}$, $Q_{0}^{(i)}$ and $Q_{-b}^{(i)}$ to describe the “up” (to level $i+1$), “within” (level $i$),  and “down” (to level $i-b$) transition rate sub-matrices in relation to the current level $i$, respectively. The block matrices $Q_{1}^{(i)} (i\leq cb-1)$, $Q_{0}^{(i)} (i\leq cb-1)$, and $Q_{-b}^{(i)}$ $(b\leq i\leq cb-1)$, are $(\lfloor \dfrac{i}{b} \rfloor+1) \times(\lfloor \dfrac{i+1}{b} \rfloor+1)$, $(\lfloor \dfrac{i}{b} \rfloor+1) \times(\lfloor \dfrac{i}{b} \rfloor+1)$, and $(\lfloor \dfrac{i}{b} \rfloor+1) \times(\lfloor \dfrac{i}{b} \rfloor)$ matrices whose contents are given by
	\begin{equation*}
		Q_{1}^{(i)}= diag(\lambda, \lambda, \ldots, \lambda),
	\end{equation*}
	
	\begin{equation*}
		Q_{0}^{(i)}=\left(\begin{array}{ccccc}
			-q_{0}^{(i)} & \lfloor \dfrac{i}{b} \rfloor \alpha & 0 & \cdots & 0 \\
			0 & -q_{1}^{(i)} & (\lfloor \dfrac{i}{b} \rfloor-1) \alpha  & \ddots & \vdots \\
			\vdots & \ddots & \ddots & \ddots & 0\\
			\vdots & \vdots & \ddots & -q_{\lfloor i/b \rfloor-1}^{(i)} & \alpha\\
			0 & \cdots & \cdots & 0 & -q_{\lfloor i/b \rfloor}^{(i)}
		\end{array}\right),
	\end{equation*}
	where $q_{j}^{(i)}=\lambda+j\mu+(\lfloor \dfrac{i}{b} \rfloor-j) \alpha$, for $j=0,1,\ldots,\lfloor \dfrac{i}{b} \rfloor$,
	\begin{equation*}
		Q_{-b}^{(i)}=\left(\begin{array}{cccc}
			0 & 0 & \cdots & 0 \\
			0 & \mu & \ddots & \vdots \\
			\vdots & \ddots & \ddots & 0\\
			\vdots & \cdots & 0 & \lfloor \dfrac{i}{b} \rfloor\mu\\
			0 & \cdots & 0 & (\lfloor \dfrac{i}{b} \rfloor+1)\mu
		\end{array}\right).
	\end{equation*}
	The block matrix $Q_{-b}^{(i)}$ $(cb\leq i < cb+b)$ is $(c+1)\times c$ matrix whose contents are given by
	\begin{equation*}
		Q_{-b}^{(i)}=\left(\begin{array}{cccc}
			0 & 0 & \cdots & 0 \\
			0 & \mu & \ddots & \vdots \\
			\vdots & \ddots & \ddots & 0\\
			\vdots & \cdots & 0 & (c-1)\mu\\
			0 & \cdots & 0 & c\mu
		\end{array}\right).
	\end{equation*}
	Furthermore, the block matrices $Q_{1}^{(i)}(i\geq cb)$, $Q_{0}^{(i)}(i\geq cb)$, and $Q_{-b}^{(i)}$ $(i\geq cb+b)$ are independent of $i$ and are given as follows: 
	\begin{equation*}
		Q_{1}^{(i)}=Q_{1}= \lambda I, \quad
		Q_{-b}^{(i)}=Q_{-b}=diag(0, \mu, \ldots, c \mu),
	\end{equation*}
	\begin{equation*}
		Q_{0}^{(i)}=Q_{0}=\left(\begin{array}{ccccc}
			-q_{0} & c\alpha & 0 & \cdots & 0 \\
			0 & -q_{1} & (c-1)\alpha  & \ddots & \vdots \\
			\vdots & \ddots & \ddots & \ddots & 0\\
			\vdots & \vdots & \ddots & -q_{c-1} & \alpha\\
			0 & \cdots & \cdots & 0 & -q_{c}
		\end{array}\right),
	\end{equation*}
	where $q_{j}=\lambda+j\mu+(c-j) \alpha$, for $j=0,1,\ldots,c$. Let
	\begin{align*}
		\boldsymbol{\pi}_{i}&=(\pi_{0,i}, \pi_{1, i}, \ldots, \pi_{min (\lfloor i/b \rfloor, c), i}), \quad \text{for } i\in \mathbb{Z_{+}}=\{0,1,2,\ldots\},\\
		\boldsymbol{\pi}&=(\boldsymbol{\pi}_{0}, \boldsymbol{\pi}_{1}, \ldots).  
	\end{align*}
	Then, the stationary distribution $\boldsymbol{\pi}$ is the unique solution of 
	\begin{align*}
		\boldsymbol{\pi} Q&= \textbf{0},\\
		\boldsymbol{\pi} \boldsymbol{e}&= 1, 
	\end{align*}
	where \textbf{0} is a row vector of zeros and $\boldsymbol{e}$ is a column vector of ones with an appropriate dimension. According to the matrix analytic method \cite{neuts1994matrix}, we have
	\begin{equation*}
		\boldsymbol{\pi}_{i}=\boldsymbol{\pi}_{i-1}R^{(i)}, \quad \text{for } i\in \mathbb{N}=\{1,2,3,\ldots\},
	\end{equation*}
	and $\boldsymbol{\pi}_{0}$ is the solution of the boundary equation and the normalization condition:
	\begin{alignat*}{3}
		\boldsymbol{\pi}_{0}(Q^{(0)}_{0}+\prod_{i=1}^{b} R^{(i)}Q^{(b)}_{-b}) &=\textbf{0}, \\  
		\boldsymbol{\pi}_{0}(I+R^{(1)}+R^{(1)}R^{(2)}+\ldots) \boldsymbol{e} &=1,
	\end{alignat*}
	where $\mathit{I}$ is the $(c + 1) \times (c + 1)$ identity matrix, and $\{ R^{(i)}; i \in \mathbb{N} \}$ is the minimal non-negative solution of
	\begin{equation*}
		Q_{1}^{(i-1)}+ R^{(i)} Q_{0}^{(i)}+\prod_{j=0}^{b} R^{(i+j)}Q_{-b}^{(i+b)}=\mathit{O}.
	\end{equation*}
	
	We further consider the finite-buffer case. With the same setting as in the M/M/$c$/SET--BATCH model above, the two-dimensional process $\{X(t)=(B(t), N(t));t\geq0\}$ forms a Markov chain in the state space 
	\begin{equation*}
		S=\{(i,j);i=0,1,\ldots,c,j=ib,ib+1,\ldots, K\}.
	\end{equation*}
	The system can be formulated as a QSF process with the infinitesimal generator $Q$ given as
	\begin{equation*}
		Q=\left(\begin{array}{cccccccccccc}
			Q_{0}^{(0)} & Q_{1}^{(0)} & \mathit{O} & \mathit{O} & \mathit{O} & \cdots & \mathit{O} \\
			\mathit{O} & Q_{0}^{(1)} & Q_{1}^{(1)} & \mathit{O} & \mathit{O} & \cdots  & \mathit{O} \\
			\vdots & \ddots & \ddots & \ddots & \ddots & \ddots & \vdots \\
			Q_{-b}^{(b)} & \mathit{O} & \cdots &  Q_{0}^{(b)} & Q_{1}^{(b)}  & \ddots & \mathit{O}\\
			\vdots & \ddots & \ddots & \ddots & \ddots & \ddots & \vdots\\
			
			\mathit{O} & \cdots & Q_{-b}^{(K)} & \cdots & \cdots & Q_{0}^{(K)} & Q_{1}^{(K)} 
		\end{array}\right) \text { , }
	\end{equation*}
	where the block matrices $Q_{1}^{(i)} (i\leq K)$, $Q_{0}^{(i)} (i\leq K)$, and $Q_{-b}^{(i)}$ $(b\leq i\leq K)$ are similar to the corresponding block matrices in the infinite-buffer case. Let
	\begin{align*}
		\boldsymbol{\pi}_{i}&=(\pi_{0,i}, \pi_{1, i}, \ldots, \pi_{min (\lfloor i/b \rfloor, c), i}), \quad \text{for } i=0,1,2,\ldots,K,\\
		\boldsymbol{\pi}&=(\boldsymbol{\pi}_{0}, \boldsymbol{\pi}_{1}, \ldots, \boldsymbol{\pi}_{K}).  
	\end{align*}
	According to the matrix analytic method \cite{neuts1994matrix}, the subvectors $\boldsymbol{\pi}_{0}, \boldsymbol{\pi}_{1}, \ldots, \boldsymbol{\pi}_{K}$ are the solution of the boundary equations and the normalization condition:
	\begin{alignat*}{3}
		\boldsymbol{\pi}_{0}Q^{(0)}_{0}+\boldsymbol{\pi}_{b}Q^{(b)}_{-b} &=\textbf{0}, \\
		\boldsymbol{\pi}_{0}Q^{(0)}_{1}+\boldsymbol{\pi}_{1}Q^{(1)}_{0}+\boldsymbol{\pi}_{b+1}Q^{(b+1)}_{-b} &=\textbf{0},\\
		& \vdotswithin{=}\\
		\boldsymbol{\pi}_{K-1}Q^{(K-1)}_{1}+\boldsymbol{\pi}_{K}Q^{(K)}_{0} &=\textbf{0},\\
		(\boldsymbol{\pi}_{0}+\boldsymbol{\pi}_{1}+ \ldots +\boldsymbol{\pi}_{K}){\boldsymbol{e}} &=1.
	\end{alignat*}
	For this case, we need to solve $K+1$ vector boundary equations, but these are reduced to $cb$ vector equations in the infinite-buffer case. Hence, to reduce the computational complexity, we approximate the finite-buffer system by the infinite-buffer system, which is analyzed in the following subsections.
	
	\subsubsection{Homogeneous part}
	It should be noted that $Q_{1}^{(i)}=Q_{1}$ $(i\geq cb)$, $Q_{0}^{(i)}=Q_{0}$ $(i\geq cb)$, and $Q_{-b}^{(i)}=Q_{-b}$ $(i\geq cb+b)$. Hence, we have $R^{(i)}=R$ for $i\geq cb+1$, where $R$ is the minimal nonnegative solution of the following equation:
	\begin{equation}
		Q_{1}+RQ_{0}+R^{b+1}Q_{-b}=\mathit{O}, \label{eq:1}
	\end{equation}
	It can be seen that $R$ is an upper triangular matrix, that is, $R(i,j)= r_{i,j}$ $(j\geq i)$ and  $R(i,j)= 0$ $(j<i)$, because $Q_{0},$ $Q_{1},$ and $Q_{-b}$ are upper triangular matrices. Similar structures for the matrix $R$ are observed in the models presented in \cite{neuts1994matrix,phung2017exact,armony2019exact,hanukov2021explicit} for $b=1$. Taking the diagonal part of (\ref{eq:1}), we have
	\begin{align}
		\lambda - (\lambda + i\mu+ (c-i)\alpha)r_{i,i}+ i\mu r^{b+1}_{i,i} = 0, \quad i=0,1,\ldots, c-1, c. \label{eq:2}
	\end{align}
	Because $R$ is the minimal nonnegative solution of (\ref{eq:1}), we have to choose the smallest root for $r_{i,i}$. 
	
	\begin{lemma} \label{lem:1}
		Eq.~(\ref{eq:2}) has a unique root satisfying $0<r_{i,i}<1$.
	\end{lemma}
	The proof of Lemma \ref{lem:1} is deferred to~\ref{Ad:1}.
	Multiplying (\ref{eq:2}) by $(i\mu)^{-1}$ and rearranging the result, we have
	\begin{align}
		r_{i,i} = p + qr^{b+1}_{i,i} \label{eq:3}
	\end{align}
	where $p=\frac{\lambda}{\lambda + i\mu+ (c-i)\alpha}$ and $q=\frac{i\mu}{\lambda + i\mu+ (c-i)\alpha}$.
	Then, we can obtain the minimal nonnegative solution of (\ref{eq:3}) from the recursive sequence
	\begin{align}
		x_{0}=0, \quad x_{n+1}= p+qx_{n}^{b+1}. \label{eq:4}
	\end{align}
	\begin{lemma}\label{lem:2}
		The recursive sequence (\ref{eq:4}) is monotonically increasing and bounded above by 1. Thus, we have $\lim_{n \to \infty}x_n=r_{i,i}$.
	\end{lemma}
	The proof of Lemma \ref{lem:2} is given in~\ref{Ad:2}.
	
	We now shift to the non-diagonal elements $r_{i,j}$ $(j>i)$. Taking the $(i,j)$ element in (\ref{eq:1}), we obtain
	\begin{align*}
		&(c-j+1)\alpha r_{i,j-1}- (\lambda+(c-j)\alpha+j\mu)r_{i,j} \\
		&+j\mu \sum_{h_{b}=i}^{j}\sum_{h_{b-1}=i}^{h_{b}}\ldots \sum_{h_{1}=i}^{h_{2}}r_{i,h_{1}}r_{h_{1},h_{2}}\ldots r_{h_{b-1},h_{b}}r_{h_{b},h_{j}}=0.
	\end{align*}
	For $j=i+1$, we have
	\begin{align*}
		&(c-i)\alpha r_{i,i}- (\lambda+(c-i-1)\alpha+(i+1)\mu)r_{i,j} \\
		&+(i+1)\mu r_{i, i+1}\sum_{k=0}^{b}r_{i,i}^{b-k}r_{i+1, i+1}^{k}=0.
	\end{align*}
	Thus,
	\begin{align*}
		&r_{i, i+1}=\frac{(c-i)\alpha r_{i,i}}{\lambda+(c-i-1)\alpha+(i+1)\mu-(i+1)\mu\sum_{k=0}^{b}r_{i,i}^{b-k}r_{i+1, i+1}^{k}}, \\ 
		&\text{for }i=0,1,\ldots, c-1.
	\end{align*}
	For the general case, we obtain
	\begin{align*}
		&r_{i, j}=\frac{(c-j+1)\alpha r_{i,j-1}+j\mu A_{r}}{\lambda+(c-j)\alpha+j\mu - j\mu\sum_{k=0}^{b}r_{i,i}^{b-k}r_{i+1, i+1}^{k}}, \quad \text{for }j > i+1,
	\end{align*}
	where 
	\begin{align*}
		A_{r}=\sum_{k=1}^{b}\sum_{h_{k}=i+1}^{j-1}\sum_{h_{k-1}=i}^{h_{k}}\ldots \sum_{h_{1}=i}^{h_{2}}r_{i,h_{1}}r_{h_{1},h_{2}}\ldots r_{h_{k},j}r_{j,j}^{b-k}.
	\end{align*}
	From these recursive formulas, we can obtain the elements of the rate matrix from the diagonal part and subsequently the upper diagonal parts.
	\subsubsection{Non-homogeneous part}\label{sec:4.1.2}
	Because $R^{(i)}=R$ $(i=cb+1, cb+2,\ldots)$, we only need to find $R^{(i)}$ $(i=cb,cb-1,\ldots, 1)$, which can be derived using the following formula:
	\begin{align*}
		R^{(i)} = -Q_{1}^{(i-1)}(Q_{0}^{(i)}+ \prod_{j=1}^{b} R^{(i+j)}Q_{-b}^{(i+b)})^{-1}, 
	\end{align*}
	for $i=cb, cb-1, \ldots, 1.$ We can obtain $R^{(i)}$ $(i=cb, cb-1,\ldots, 1)$ from the following system of linear equations:
	\begin{align*}
		R^{(i)}(Q_{0}^{(i)}+\prod_{j=1}^{b} R^{(i+j)}Q_{-b}^{(i+b)}) = -Q_{1}^{(i-1)}, 
	\end{align*}
	for $i=cb, cb-1, \ldots, 1.$ 
	Because the rate matrices are upper trapezoidal matrices, this system of linear equations can be efficiently solved as follows. We first need to solve the following equation:
	\begin{align*}
		XA = -Q_{1}^{(i-1)},
	\end{align*}
	where $A= Q_{0}^{(i)}+\prod_{j=1}^{b} R^{(i+j)}Q_{-b}^{(i+b)}$. $A$ and $X$ are upper trapezoidal matrices of size and $(\lfloor\frac{i}{b}\rfloor+1)\times(\lfloor\frac{i}{b}\rfloor+1)$ and $(\lfloor\frac{i-1}{b}\rfloor+1)\times(\lfloor\frac{i}{b}\rfloor+1)$, respectively. Let $x_{j}=(0, 0,\ldots, x_{j,j}, x_{j, j+1}, \ldots, x_{j, \lfloor\frac{i}{b}\rfloor})$ $(j= 0,1,\ldots, \lfloor\frac{i-1}{b}\rfloor)$ denote the $j$th row vector of $X$. The above equation is equivalent to
	\begin{align*} 
		x_{j}A= (0, 0, \ldots, -\lambda, 0, \ldots, 0), \quad j= 0,1,\ldots, \lfloor\frac{i-1}{b}\rfloor,
	\end{align*}
	where the $-\lambda$ is the $(j+1)$th entry of the vector in the right hand side. The solution of this equation is given as follows:
	\begin{align}\label{eq5}
		x_{j,j}= -\frac{\lambda}{a_{j,j}}, \quad x_{j,h}= -\frac{\sum_{k=j}^{h-1}x_{j,k}a_{k,h}}{a_{h,h}}, 
	\end{align}
	for $h= j+1, j+2,\ldots, \lfloor\frac{i}{b}\rfloor,$ where $a_{i,j}$ is the $(i, j)$ entry of $A$.
	\subsection{ M/M/$c$/SET--BATCH/Delayoff}\label{sec:4.2}
	The Markov chain for M/M/$c$/SET--BATCH/Delayoff is shown in Fig.~\ref{Fig2}. Similar to the M/M/$c$/SET--BATCH model, the system stability condition is $\lambda < cb\mu$. Let $A(t)$ and $N(t)$ denote the number of busy or idle servers and the number of jobs in the system at time $t$, respectively. In the current setting, the two-dimensional process $\{X(t)=(A(t), N(t));t\geq0\}$ forms a Markov chain in the state space 
	\begin{equation*}
		S=\{(i,j);i=0,1,\ldots,c,j=0,1,2,\ldots\}.
	\end{equation*}
	The system can be formulated as a Quasi-Skip-Free process (QSF) with the infinitesimal of $\{X(t)\}$ given as
	\begin{equation*}
		Q=\left(\begin{array}{cccccccccccc}
			Q_{0}^{(0)} & Q_{1} & \mathit{O} & \mathit{O} & \mathit{O} & \mathit{O} & \cdots \\
			\mathit{O} & Q_{0}^{(1)} & Q_{1} & \mathit{O} & \mathit{O} & \mathit{O} & \cdots \\
			\vdots & \ddots & \ddots & \ddots & \ddots & \ddots & \cdots \\
			Q_{-b} & \mathit{O} & \cdots &  Q_{0}^{(b)} & Q_{1}  & \ddots & \cdots \\
			\mathit{O} & Q_{-b} & \mathit{O} & \cdots & Q_{0}^{(b+1)} & Q_{1} & \cdots \\
			\vdots & \vdots & \vdots & \vdots & \vdots & \vdots & \ddots 
		\end{array}\right) \text { , }
	\end{equation*}
	where $\mathit{O}$ denotes the zero matrix with an appropriate dimension. The block matrix $Q_{0}^{(i)} (i<b)$ is $(c+1) \times(c+1)$ matrix whose contents are given by
	\begin{equation*}
		Q_{0}^{(i)} = \left(\begin{array}{ccccc}
			-q_{0} & 0 & \cdots & 0 \\
			\beta & -q_{1} & \ddots & \vdots \\
			\vdots & \ddots & \ddots & \vdots\\
			0 & \cdots & c\beta & -q_{c}
		\end{array}\right),
	\end{equation*}
	where $q_{j}=\lambda+j\beta$, for $j=0,1,\ldots,c$.
	
	The block matrix $Q_{0}^{(i)} (i\geq b)$ is $(c+1) \times(c+1)$ matrix whose contents are given by
	\begin{equation*}
		Q_{0}^{(i)}=\left(\begin{array}{ccccc}
			-q_{0}^{(i)} & m_{0}^{(i)} & 0 & \cdots & 0 \\
			n_{1}^{(i)} & -q_{1}^{(i)} & m_{1}^{(i)}  & \ddots & \vdots \\
			\vdots & \ddots & \ddots & \ddots & \vdots\\
			\vdots & \ddots & n_{c-1}^{(i)} & -q_{c-1}^{(i)} & m_{c-1}^{(i)}\\
			0 & \cdots & 0 & n_{c}^{(i)} & -q_{c}^{(i)}
		\end{array}\right),
	\end{equation*}
	where $m_{j}^{(i)} =\text{max}(\lfloor \dfrac{i}{b} \rfloor-j, 0)\alpha$, for $j=0,1,\ldots,c-1$, \quad $n_{j}^{(i)}=\text{max}(j-\lfloor \dfrac{i}{b} \rfloor, 0) \beta$, for $j=1,\ldots,c$, 
	and \quad $q_{j}^{(i)}=\lambda+j\mu+m_{j}^{(i)}+n_{j}^{(i)}$, for $j=0,1,\ldots,c$.
	
	Furthermore, the block matrices $Q_{1}$, $Q_{-b}$, and $Q_{0}^{(i)} (i\geq cb)$ are independent of $i$ and are given as follows:
	\begin{equation*}
		Q_{1}= \lambda I,\quad
		Q_{-b}= diag(0, \mu, \ldots, c\mu), 
	\end{equation*}
	\begin{equation*}
		Q_{0}^{(i)}=Q_{0}=\left(\begin{array}{ccccc}
			-q_{0} & c\alpha & 0 & \cdots & 0 \\
			0 & -q_{1} & (c-1)\alpha  & \ddots & \vdots \\
			\vdots & \ddots & \ddots & \ddots & 0\\
			\vdots & \cdots & \ddots & -q_{c-1} & \alpha\\
			0 & \cdots & \cdots & 0 & -q_{c}
		\end{array}\right),
	\end{equation*}
	where $q_{j}=\lambda+j\mu+(c-j) \alpha$, for $j=0,1,\ldots,c$. Let
	\begin{align*}
		\boldsymbol{\pi}_{i}&=(\pi_{0,i}, \pi_{1, i}, \ldots, \pi_{c, i}), \quad \text{for } i\in \mathbb{Z_{+}}=\{0,1,2,\ldots\},\\
		\boldsymbol{\pi}&=(\boldsymbol{\pi}_{0}, \boldsymbol{\pi}_{1}, \ldots).  
	\end{align*}
	Then, the stationary distribution $\boldsymbol{\pi}$ is the unique solution of the following equations:
	\begin{align*}
		\boldsymbol{\pi} Q&= \textbf{0}, \\
		\boldsymbol{\pi} \boldsymbol{e}&= 1, 
	\end{align*}
	where \textbf{0} is a row vector of zeros and $\boldsymbol{e}$ is a column vector of ones with an appropriate dimension. According to the matrix analytic method \cite{neuts1994matrix}, we have
	\begin{equation*}
		\boldsymbol{\pi}_{i}=\boldsymbol{\pi}_{cb}R^{i-cb}, \quad \text{for } i\geq cb,
	\end{equation*}
	and the subvectors $\boldsymbol{\pi}_{0}, \boldsymbol{\pi}_{1}, \ldots, \boldsymbol{\pi}_{cb}$ are the solution of the boundary equations and the normalization condition:
	\begin{alignat*}{3}
		\boldsymbol{\pi}_{0}Q^{(0)}_{0}+\boldsymbol{\pi}_{b}Q_{-b} &=\textbf{0}, \\
		\boldsymbol{\pi}_{0}Q_{1}+\boldsymbol{\pi}_{1}Q^{(1)}_{0}+\boldsymbol{\pi}_{b+1}Q_{-b} &=\textbf{0},\\
		& \vdotswithin{=}\\
		\boldsymbol{\pi}_{cb-1}Q_{1}+\boldsymbol{\pi}_{cb}(Q^{(cb)}_{0}+R^{b}Q_{-b}) &=\textbf{0},\\
		(\sum_{i=0}^{cb-1}\boldsymbol{\pi}_{i}+\boldsymbol{\pi}_{cb}(I-R)^{-1})\boldsymbol{e} &=1,
	\end{alignat*}
	where $R$ is the minimal non-negative solution of 
	\begin{equation*}
		Q_{1}+ R Q_{0}+R^{b+1} Q_{-b}=\mathit{O}.
	\end{equation*}
	It should be noted that this equation is the same as (\ref{eq:1}).
	
	\subsection{ M/M/$c$/SET--BATCH/Dynamic}\label{sec:4.3}
	The M/M/$c$/SET--BATCH/Dynamic chain shown in Fig.~\ref{Fig3} is analyzed similarly to the M/M/$c$/SET--BATCH. The system stability condition is $\lambda < cb_2\mu$. Let $B(t)$ and $N(t)$ denote the number of busy servers and the number of jobs in the queue at time $t$, respectively. In the current setting, the two-dimensional process $\{X(t)=(B(t), N(t));t\geq0\}$ forms a Markov chain in the state space 
	\begin{equation*}
		S=\{(i,j);i=0,1,\ldots,c,j=0,1,2,\ldots\}.
	\end{equation*}
	The infinitesimal generator $Q$ of $\{X(t)\}$ is given as
	\begin{equation*}
		Q=\left(\begin{array}{cccccccccccc}
			Q_{0}^{(0)} & Q_{1} & \mathit{O} & \cdots & \cdots & \cdots &  \cdots & \cdots & \cdots \\
			\mathit{O} & Q_{0}^{(1)} & Q_{1} & \mathit{O} & \cdots & \cdots & \cdots & \cdots & \cdots \\
			\vdots & \ddots & \ddots & \ddots & \ddots & \cdots & \cdots & \cdots & \cdots \\
			Q_{-b_1}^{(b_1)} & \mathit{O} & \cdots &  Q_{0}^{(b_1)} & Q_{1}  & \mathit{O} & \cdots & \cdots & \cdots \\
			\vdots & \vdots & \ddots & \ddots & \ddots & \ddots & \ddots & \cdots & \cdots \\
			Q_{-b_2}^{(b_2)} & \mathit{O} & \mathit{O} & \cdots & \cdots & Q_{0}^{(b_2)} & Q_{1} & \mathit{O} & \cdots \\
			\mathit{O} & Q_{-b_2}^{(b_2+1)} & \mathit{O} & \cdots & \cdots & \cdots & Q_{0}^{(b_2+1)} & Q_{1} & \cdots \\
			\vdots & \vdots & \vdots & \vdots & \vdots & \vdots & \vdots & \vdots & \ddots 
		\end{array}\right) \text { , }
	\end{equation*}
	where $\mathit{O}$ denotes the zero matrix with an appropriate dimension. The block matrix $Q_{0}^{(i)} (i \leq b_1-1)$ is $(c+1) \times(c+1)$ matrix whose contents are given by
	\begin{equation*}
		Q_{0}^{(i)} = \left(\begin{array}{ccccc}
			-\lambda & 0 & \cdots & 0  \\
			\mu & -(\lambda+\mu) & \ddots & 0 \\
			\vdots & \ddots & \ddots & \ddots\\
			0 & \cdots & c\mu & -(\lambda+c\mu)
		\end{array}\right).
	\end{equation*}
	The block matrices $Q_{-i}^{(i)}$ $(b_1 \leq i< b_2)$ and $Q_{0}^{(i)}$ $(b_1 \leq i < b_2)$ are $(c+1) \times(c+1)$ matrices whose contents are given by
	\begin{equation*}
		Q_{-i}^{(i)} = \left(\begin{array}{ccccc}
			0 & \alpha & 0 & \cdots  \\
			0 & \mu & \alpha & \ddots \\
			\vdots & \ddots & \ddots & \ddots\\
			0 & \cdots & 0 & c\mu
		\end{array}\right),
	\end{equation*}
	\begin{equation*}
		Q_{0}^{(i)} = \left(\begin{array}{ccccc}
			-(\lambda+\alpha) & 0 & \cdots & 0  \\
			0 & -(\lambda+\mu+\alpha) & \ddots & \vdots \\
			\vdots & \ddots & \ddots & \vdots\\
			0 & \cdots & 0 & -(\lambda+c\mu)
		\end{array}\right),
	\end{equation*}
	
	The block matrices $Q_{-b_2}^{(i)}$ $(i \geq b_2)$ and $Q_{0}^{(i)}$ $( b_1 \leq i < b_1+(c-1)b_2)$ are $(c+1) \times(c+1)$ matrices whose contents are given by
	\begin{equation*}
		Q_{-b_2}^{(i)} = \left(\begin{array}{ccccc}
			0 & q_{0} & 0 & \cdots  \\
			0 & \mu & q_{1} & \ddots \\
			\vdots & \ddots & \ddots & \ddots\\
			0 & \cdots & 0 & c\mu
		\end{array}\right),
	\end{equation*}
	where $q_{j}= \text{min}(\lfloor\frac{i}{b_2}\rfloor + \text{min}(\lfloor\frac{i \mod b_2}{b_1}\rfloor, 1),c-j) \alpha$, for $j=0,1,\ldots, c-1.$
	\begin{equation*}
		Q_{0}^{(i)} = \left(\begin{array}{ccccc}
			-q_{0} & 0 & \cdots & \cdots  \\
			0 & -q_{1} & 0 & \ddots \\
			\vdots & \ddots & \ddots & \ddots\\
			0 & \cdots & 0 & -q_{c}
		\end{array}\right),
	\end{equation*}
	where $q_j= \lambda+ j\mu+ \text{min}(\lfloor\frac{i}{b_2}\rfloor + \text{min}(\lfloor\frac{i \mod b_2}{b_1}\rfloor, 1),c-j) \alpha$.
	
	Moreover, the block matrices $Q_{1}$, $Q_{-b_2}$ $(i \geq b_1+(c-1)b_2$), and $Q_{0}^{(i)}$ $(i\geq b_1+(c-1)b_2)$ are independent of $i$ and are given as follows:
	\begin{equation*}
		Q_{1}= \lambda I,
	\end{equation*}
	\begin{equation*}
		Q_{-b_2}^{(i)} = Q_{-b_2}=\left(\begin{array}{ccccc}
			0 & c\alpha & 0 & \cdots & 0 \\
			0 & \mu & (c-1)\alpha  & \ddots & \vdots \\
			\vdots & \ddots & \ddots & \ddots & 0\\
			\vdots & \cdots & \ddots & (c-1)\mu & \alpha\\
			0 & \cdots & \cdots & 0 & c\mu
		\end{array}\right),
	\end{equation*}
	\begin{equation*}
		Q_{0}^{(i)}=Q_{0}=\left(\begin{array}{ccccc}
			-(\lambda+c\alpha) & 0 & \cdots & 0 \\
			0 & -(\lambda+(c-1)\alpha+\mu) & \ddots & \vdots \\
			\vdots & \ddots & \ddots & 0\\
			0 & \cdots & 0 & -(\lambda+c\mu)
		\end{array}\right).
	\end{equation*}
	Let
	\begin{align*}
		\boldsymbol{\pi}_{i}&=(\pi_{0,i}, \pi_{1, i}, \ldots, \pi_{c, i}), \quad \text{for } i\in \mathbb{Z_{+}}=\{0,1,2,\ldots\},\\
		\boldsymbol{\pi}&=(\boldsymbol{\pi}_{0}, \boldsymbol{\pi}_{1}, \ldots).  
	\end{align*}
	Then, the stationary distribution $\boldsymbol{\pi}$ is the unique solution of 
	\begin{align}
		\boldsymbol{\pi} Q&= \textbf{0}, \label{eq:7}\\
		\boldsymbol{\pi} \boldsymbol{e}&= 1, \label{eq:8}
	\end{align}
	where \textbf{0} is a row vector of zeros and $\boldsymbol{e}$ is a column vector of ones with an appropriate dimension. According to the matrix analytic method \cite{neuts1994matrix}, we have
	\begin{equation*}
		\boldsymbol{\pi}_{i}=\boldsymbol{\pi}_{b_1+(c-1)b_2}R^{i-b_1-(c-1)b_2}, \quad \text{for } i\geq b_1+(c-1)b_2,
	\end{equation*}
	where $R$ is the minimal non-negative solution of
	\begin{equation*}
		Q_{1}+ R Q_{0}+R^{b_2+1} Q_{-b_2}=\mathit{O}.
	\end{equation*}
	It should be noted that the rate matrix $R$ can be solved using similar formulas as in (\ref{eq:2}). Similar to the M/M/$c$/SET--BATCH/Delayoff model, once we have determined $R$, the subvectors $\boldsymbol{\pi}_{0}, \boldsymbol{\pi}_{1}, \ldots, \boldsymbol{\pi}_{b_1+(c-1)b_2}$ can be solved by using the first $(b_1+(c-1)b_2+1)$ equations of (\ref{eq:7}) and the normalizing equation (\ref{eq:8}).
	
	\section{Performance measures}\label{sec5}
	
	\subsection{M/M/$c$/SET--BATCH}\label{sec:5.1}
	Let ${\pi}_{i}$ denote the stationary probability that there are $i$ active servers, that is, ${\pi}_{i}= \sum_{j=ib}^{\infty}\pi_{i,j}.$ Let $B$ and $S$ denote the number of busy servers and setup servers in the system in the steady state, respectively. Then, the average number of busy servers E$[B]$ and setup servers E$[S]$ are given by
	\begin{align*}
		\text{E}[B]&= \sum_{i=1}^{c}i{\pi}_{i},\\
		\text{E}[S]&= \sum_{i=0}^{c} \sum_{j=ib+b}^{\infty}\text{min}(\lfloor\frac{j-ib}{b}\rfloor,c-i)\pi_{i,j}.
	\end{align*}
	Let $Sr$ be the number of completed switches from OFF
	to ON per unit time. Thus, the mean number of switches from OFF to ON per unit time in the steady state E$[Sr]$ is as follows:
	\begin{align*}
		\text{E}[Sr]&= \sum_{i=0}^{c-1} \sum_{j=ib+b}^{\infty}\text{min}(\lfloor\frac{j-ib}{b}\rfloor,c-i)\alpha\pi_{i,j}.
	\end{align*}
	Let $N$ denote the number of jobs in the system. Then, the mean number of customers in the system E$[N]$ is given by
	\begin{equation*}
		\text{E}[N]=\sum_{i=0}^{c} \sum_{j=ib}^{\infty}j\pi_{i,j}.
	\end{equation*}
	Let $W$ denote the total time that a job spends in the system (the response time). Then, from Little's law, the mean response time E$[W]$ can be obtained by
	\begin{equation*}
		\text{E}[W]=\frac{\text{E}[N]}{\lambda}.
	\end{equation*}
	The mean power consumption is
	\begin{equation*}
		Cost=\mu\left(C_b\text{E}[B]+C_s\text{E}[S]\right).
	\end{equation*}
	Here $C_b$ represents the energy consumed per batch (in joules),
	so that $\mu C_b$ is the mean power consumption of a busy server. For $C_b$, we use the linear cost function in \cite{inoue2021queueing} for GPU-based inference servers (Tesla V100). The cost function agrees closely with the measurement results of the computing performance for Tesla V100 (the coefficient of determination $R^{2}\simeq 0.99978$). We assume that a server in setup consumes the same mean
	power as a busy server, so that $C_s=C_b$. Hence,
	\begin{align*}
		C_b= C_s= 0.04419b+0.15503.
	\end{align*}
	
	Furthermore, to evaluate the trade-off between energy consumption and performance, we use the energy--delay product (EDP) metric, also known as the energy--response time product (ERP) in \cite{gandhi2010optimality}. The ERP is given by
	\begin{equation*}
		\text{ERP}=\text{E}[W]\times Cost.
	\end{equation*}
	
	\subsection{M/M/$c$/SET--BATCH/Delayoff}
	Let ${\pi}_{i}$ denote the stationary probability that there are $i$ busy or idle servers, that is, ${\pi}_{i}= \sum_{j=0}^{\infty}\pi_{i,j}.$ Let $I$, $B$, and $S$ denote the number of idle servers, busy servers, and setup servers in the system, respectively. Similar to the M/M/$c$/SET--BATCH, the average number of idle servers E$[I]$, busy servers E$[B]$, setup servers E$[S]$, and the mean number of switches per unit time E$[Sr]$ are as follows:	
\begin{align*}
	\text{E}[I]
	&= \sum_{i=1}^{c}\sum_{j=0}^{ib-1}
	\left(i-\lfloor\frac{j}{b}\rfloor\right)\pi_{i,j},\\
	\text{E}[B]
	&= \sum_{i=1}^{c}i\pi_i
	- \sum_{i=1}^{c}\sum_{j=0}^{ib-1}
	\left(i-\lfloor\frac{j}{b}\rfloor\right)\pi_{i,j},\\
	\text{E}[S]
	&= \sum_{i=0}^{c}\sum_{j=ib+b}^{\infty}
	\min\left(\lfloor\frac{j-ib}{b}\rfloor,c-i\right)
	\pi_{i,j},\\
	\text{E}[Sr]
	&= \sum_{i=0}^{c-1}\sum_{j=ib+b}^{\infty}
	\min\left(\lfloor\frac{j-ib}{b}\rfloor,c-i\right)
	\alpha\pi_{i,j}.
\end{align*}
	Moreover, the mean number of customers in the system E$[N]$ and the mean response time E$[W]$ are given by
	\begin{align*}
		\text{E}[N]= \sum_{i=0}^{c} \sum_{j=0}^{\infty}j\pi_{i,j}, \quad
		\text{E}[W]= \frac{\text{E}[N]}{\lambda}.
	\end{align*}
	The mean power consumption is 
	\begin{equation*}
		Cost=\mu\left(
		C_b\text{E}[B]+C_s\text{E}[S]+C_i\text{E}[I]
		\right),
	\end{equation*}
	where $C_b$ and $C_s$ are the same as in the
	M/M/$c$/SET--BATCH model. Here, $\mu C_i$ is the power consumption of an idle server. We assume that idle-server power is 60\% of mean busy-server
	power, i.e., $C_i=0.6C_b$. Furthermore, the ERP is given by
	\begin{equation*}
		\text{ERP}=\text{E}[W]\times Cost.
	\end{equation*}
	
	\subsection{M/M/$c$/SET--BATCH/Dynamic}
	Let ${\pi}_{i}$ denote the stationary probability that there are $i$ busy servers, that is, ${\pi}_{i}= \sum_{j=0}^{\infty}\pi_{i,j}.$ The average number of busy servers E$[B]$ and setup servers E$[S]$, and the mean number of switches per unit time E$[Sr]$ are as follows:
\begin{align*}
	\text{E}[B]&= \sum_{i=1}^{c}i\pi_i,\\
	\text{E}[S]&= \sum_{i=0}^{c}\sum_{j=b_1}^{\infty}
	\min(\lfloor\frac{j}{b_2}\rfloor
	+\min(\lfloor\frac{j\bmod b_2}{b_1}\rfloor,1),c-i)
	\pi_{i,j},\\
	\text{E}[Sr]&= \sum_{i=0}^{c-1}\sum_{j=b_1}^{\infty}
	\min(\lfloor\frac{j}{b_2}\rfloor
	+\min(\lfloor\frac{j\bmod b_2}{b_1}\rfloor,1),c-i)
	\alpha\pi_{i,j}.
\end{align*}
	Let $L$ denote the number of jobs in the queue. Then, the mean number of jobs in the queue E$[L]$ is given by
	\begin{align*}
		\text{E}[L]= \sum_{i=0}^{c} \sum_{j=0}^{\infty}j\pi_{i,j}.
	\end{align*}
	Then from Little's law, the mean response time E[$W$] is given by
	\begin{align*}
		\text{E}[W]= \frac{\text{E}[L]}{\lambda} +\frac{1}{\mu}.
	\end{align*}
	The mean energy consumed per batch $C_b$ and the corresponding
	setup cost coefficient $C_s$ depend on the batch size, which is
	dynamic in the M/M/$c$/SET--BATCH/Dynamic model. Hence, compared with existing models in the literature, an extra step is needed to derive these costs. Let $w(i,j)$ denote the rate at which a new batch begins service at state $(i,j)$, for $j\geq b_1$, where
	\begin{equation*}
	w(i,j)=
	\alpha\min\left(
	\lfloor\frac{j}{b_2}\rfloor
	+\min\left(
	\lfloor\frac{j\bmod b_2}{b_1}\rfloor,1
	\right),c-i
	\right)+i\mu.
	\end{equation*}
	We determine the average batch size at service initiation epochs
	as follows:
	\begin{equation*}
		\text{E}[b]=
		\frac{
			\displaystyle
			\sum_{i=0}^{c}\sum_{j=b_1}^{b_2-1}
			j\pi_{i,j}w(i,j)
			+
			\sum_{i=0}^{c}\sum_{j=b_2}^{\infty}
			b_2\pi_{i,j}w(i,j)
		}{
			\displaystyle
			\sum_{i=0}^{c}\sum_{j=b_1}^{b_2-1}
			\pi_{i,j}w(i,j)
			+
			\sum_{i=0}^{c}\sum_{j=b_2}^{\infty}
			\pi_{i,j}w(i,j)
		}.
	\end{equation*}
The original expression for E$[b]$ in the published
version (DOI: 10.1016/j.peva.2025.102468) contains a typographical error in its denominator.
The correct denominator for that expression is
$1-\sum_{i=0}^{c}\sum_{j=0}^{b_1-1}\pi_{i,j}$.
The code used to generate the published figures already
used this correct denominator; therefore, the typographical
error did not affect the published numerical results.

The original expression gives a time-stationary average of the size of the batch that would be selected for service, conditional on $j\geq b_1$, whereas the revised expression gives the average batch size at service initiation epochs. We recomputed and updated all numerical results involving E$[b]$ using the revised expression for the parameter
settings considered in the paper. The resulting changes in the plotted curves are negligible. In particular, the comparisons between models in
Figs.~\ref{fig10} and ~\ref{fig11} remain unchanged. These updates do not affect the conclusions of the published version.

As in the M/M/$c$/SET--BATCH model, we have
	\begin{align*}
		C_b= C_s= 0.04419 \text{E}[b]+0.15503.
	\end{align*}
Finally, the mean power consumption $Cost$ and the ERP are derived exactly in the same way as in the M/M/$c$/SET--BATCH model.
	
\section{Numerical examples and simulations}\label{sec6}
\subsection{Computational complexity comparison}
In this section, we consider the effectiveness of our recursive algorithm for the rate matrices $R^{(i)} (i=1,2,\ldots,cb)$ in the non-homogeneous part of the M/M/$c$/SET--BATCH model. We compare the computational complexity of obtaining the rate matrices $R^{(i)}$ by (\ref{eq5}) with that of direct matrix inversion. In addition, the analysis of the M/M/$c$/SET--BATCH queue is used as an approximation for the finite-buffer case. Thus, we estimate the efficiency of our approximation compared with the exact solution for the finite-buffer case in terms of time and accuracy. We use CPU: 11th Gen Intel(R) Core(TM) i7-1185G7 @ 3.00 GHz, Windows 11 system, Python 3.10.9, to test the running time.

\begin{table}[htbp]
	\centering
	\caption{Computational complexity of calculating
		$R^{(i)}$ $(i=1,2,\ldots,cb)$ with our method
		($R$-structure) and the direct matrix inversion method.}
	\label{tb:1}
	
	\renewcommand{\arraystretch}{1.1}
	\begin{tabular}{@{}lrrrr@{}}
		\toprule
		\multicolumn{5}{l}{Time (seconds)} \\
		\midrule
		Number of servers
		& \multicolumn{2}{l}{$b=2$}
		& \multicolumn{2}{l}{$b=3$} \\
		\cmidrule(lr){2-3}\cmidrule(l){4-5}
		& $R$-structure & Direct inversion
		& $R$-structure & Direct inversion \\
		\midrule
		5   & 0.001  & 0.002   & 0.001   & 0.003    \\
		10  & 0.002  & 0.023   & 0.003   & 0.045    \\
		50  & 0.636  & 5.676   & 0.778   & 10.408   \\
		70  & 1.863  & 12.206  & 3.005   & 37.102   \\
		100 & 7.839  & 48.915  & 9.270   & 108.055  \\
		200 & 79.090 & 761.552 & 103.301 & 1657.326 \\
		\bottomrule
	\end{tabular}
\end{table}
	
	Table \ref{tb:1} shows the running times needed to calculate the matrices $R^{(i)}$ in the M/M/$c$/SET--BATCH model with batch size $b$ of 2 and 3. We observe that our method has a much shorter running time than the direct matrix inversion method, especially when the batch size and number of servers are large. Specifically, the direct inversion method takes approximately nine times longer than our method for $b=2$ and about 16 times longer for $b=3$, when $c=200$. Therefore, the larger the number of servers and the batch size, the more efficient our method is.
	
\begin{table}[htbp]
	\centering
	\caption{Computational complexity and accuracy: exact and
		approximate methods for finite-buffer cases
		($b=2, c=5, \rho=0.7$).}
	\label{tb:2}
	
	\renewcommand{\arraystretch}{1.1}
	\begin{tabular}{@{}lrrr@{}}
		\toprule
		\multirow{2}{*}{Buffer $K$}
		& \multicolumn{2}{c}{Time (seconds)}
		& \multirow{2}{*}{\shortstack{Percentage Error\\of E$[N]$ (\%)}} \\
		\cmidrule(lr){2-3}
		& Exact method & Approximation & \\
		\midrule
		20   & 0.006  & \multirow{5}{*}{0.002} & 26.307 \\
		50   & 0.011  &                       & 0.618  \\
		100  & 0.058  &                       & 0.496  \\
		500  & 4.745  &                       & 0.496  \\
		1000 & 35.946 &                       & 0.496  \\
		\bottomrule
	\end{tabular}
\end{table}
	
	Tables \ref{tb:2} to \ref{tb:4} compare the running times and the percentage errors in the performance measure E$[N]$ for our approximation and the exact method, for the finite-buffer case. Our approximation method is more efficient than the exact method and has a relatively good accuracy (error below 1\%) when the buffer size $K$ is not too small. We observe from Tables \ref{tb:3} and \ref{tb:4} that as we increase the number of servers, the running time of the exact method increases dramatically while the running time of our approximate method increases insignificantly, and the percentage error decreases. Specifically, when $K=1000$, the running time of our approximation is 16,000 times faster than the exact method while the error is only about 0.1\%, in Table \ref{tb:4}. Thus, as the number of servers and batch size increase, both the efficiency and accuracy of our approximation improves.
	
\begin{table}[htbp]
	\centering
	\caption{Computational complexity and accuracy: exact and
		approximate methods for finite-buffer cases
		($b=2, c=20, \rho=0.7$).}
	\label{tb:3}
	
	\renewcommand{\arraystretch}{1.1}
	\begin{tabular}{@{}lrrr@{}}
		\toprule
		\multirow{2}{*}{Buffer $K$}
		& \multicolumn{2}{c}{Time (seconds)}
		& \multirow{2}{*}{\shortstack{Percentage Error\\of E$[N]$ (\%)}} \\
		\cmidrule(lr){2-3}
		& Exact method & Approximation & \\
		\midrule
		50   & 0.101    & \multirow{5}{*}{0.049} & 8.732 \\
		100  & 0.627    &                       & 0.214 \\
		200  & 5.784    &                       & 0.223 \\
		500  & 89.123   &                       & 0.223 \\
		1000 & 1990.714 &                       & 0.223 \\
		\bottomrule
	\end{tabular}
\end{table}

\begin{table}[htbp]
	\centering
	\caption{Computational complexity and accuracy: exact and
		approximate methods for finite-buffer cases
		($b=3, c=20, \rho=0.7$).}
	\label{tb:4}
	
	\renewcommand{\arraystretch}{1.1}
	\begin{tabular}{@{}lrrr@{}}
		\toprule
		\multirow{2}{*}{Buffer $K$}
		& \multicolumn{2}{c}{Time (seconds)}
		& \multirow{2}{*}{\shortstack{Percentage Error\\of E$[N]$ (\%)}} \\
		\cmidrule(lr){2-3}
		& Exact method & Approximation & \\
		\midrule
		70   & 0.216    & \multirow{5}{*}{0.080} & 9.701 \\
		100  & 0.519    &                       & 0.601 \\
		200  & 5.175    &                       & 0.119 \\
		500  & 80.344   &                       & 0.119 \\
		1000 & 2139.136 &                       & 0.119 \\
		\bottomrule
	\end{tabular}
\end{table}
	
\subsection{M/M/$c$/SET--BATCH}\label{sec6.1}
In this section, we use simulations to investigate the effect
of the setup time distribution on the system performance.
Then, we present the performance measures based on our
analytical results to give insights into the batching policy
in the exponential setup time model.
In all the numerical examples, we define the traffic intensity
$\rho=\lambda/(cb\mu)$ and fix the mean service time
$1/\mu=1\,\mathrm{s}$. Under this parameter setting, including
	the factor $\mu$ in the cost expressions does not change
	the numerical values reported in the published version.
	
	\subsubsection{Effect of the setup time distribution}
	We first observe how the setup time distribution affects the performance of the M/M/$c/$SET--BATCH system. We fix the batch size $b=2,$ the number of servers $c=20$ , and the mean setup time at 10 times the mean service time, that is, $1/\alpha = 10 \text{s}$. We present analytical results for the exponential setup time model and simulations for three additional models: the deterministic setup time model, the Erlang setup time model (with shape parameter $k=9$), and the hyperexponential setup time model (with coefficient of variation $\text{CV}=1.5$, a mixture of two exponential distributions). Although all setup time distributions maintain the same mean of 10, the coefficient of variation differs. 
	
	Figs.~\ref{fig4a} and \ref{fig4b} represent the mean response time and the power consumption against the traffic intensity for four cases: exponentially distributed, hyperexponentially distributed, Erlang distributed, and deterministic setup times. We observe from Fig.~\ref{fig4a} that the mean response time decreases with $\rho$ under a light traffic regime, but increases with $\rho$ in a heavy traffic regime. This is because increasing the traffic intensity will increase the number of setup servers in a light traffic regime, reducing the mean response time. Nevertheless, in a heavy traffic regime, almost all servers are always active; therefore, the mean response time increases with the traffic intensity, as in the M/M/$c$ systems without setup time. In addition, we observe in Fig.~\ref{fig4b} that the power consumption increases with the traffic intensity for all settings of the setup time distributions as expected. 
	
	The mean response time and the power consumption decrease as the coefficient of variation CV increases and are bounded from below by the hyperexponential curve. The M/M/$c/$SET--BATCH models with deterministic and Erlang setup times exhibit higher latency and power consumption than the corresponding models with exponential and especially hyperexponential setup times. In other words, the performance of the setup models with the batching policy is better when the setup time has a larger coefficient of variation.
	\begin{figure}
		\centering
		\begin{subfigure}{0.49\linewidth}
			\includegraphics[width=\linewidth]{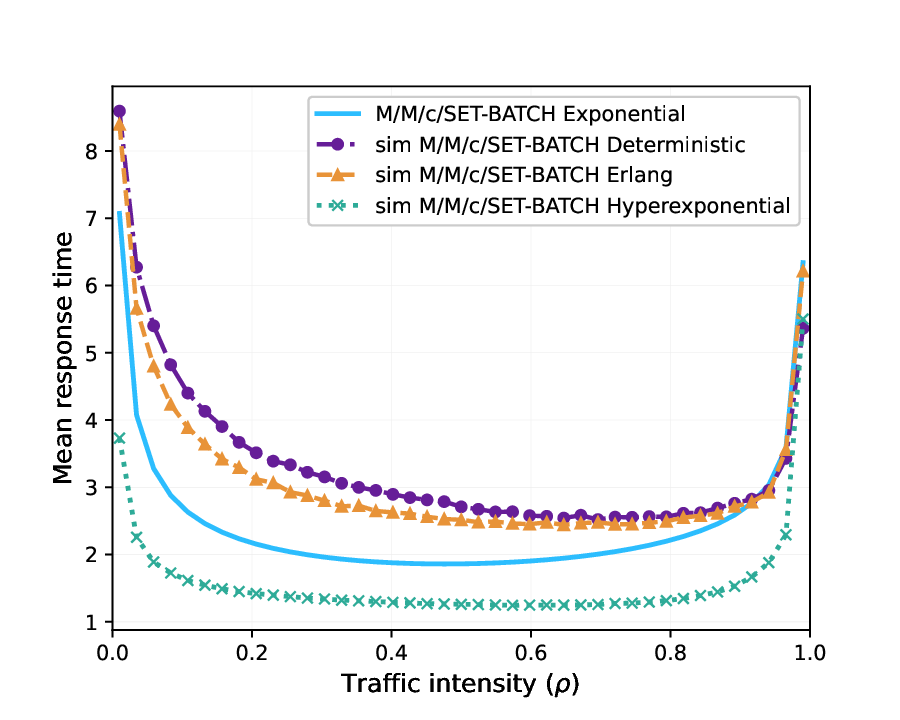}
			\caption{E$[W]$ (s)}
			\label{fig4a}
		\end{subfigure}
		\begin{subfigure}{0.49\linewidth}
			\includegraphics[width=\linewidth]{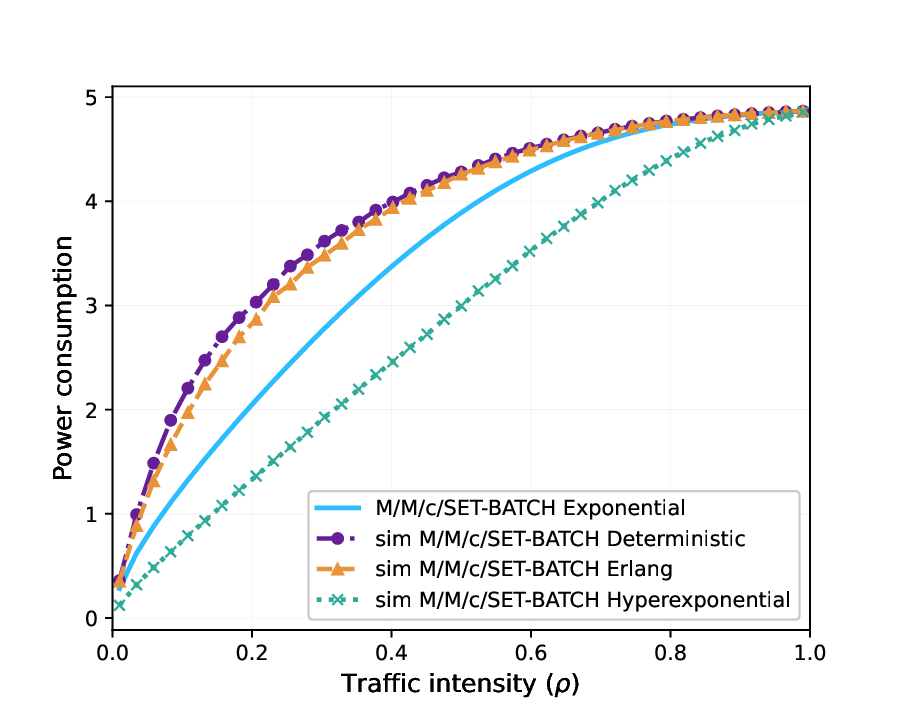}
			\caption{$Cost$ (W)}
			\label{fig4b}
		\end{subfigure}
		\caption{ The effect of the setup time distribution on system performance.}
		\label{fig4}
	\end{figure}
	
	\subsubsection{Effect of the batching policy}
	To analyze the efficiency of the batching policy in the M/M/$c/$SET--BATCH queue (BATCH model), we compare it with the M/M/$c$/Setup model, which operates without the batching policy (w/o BATCH model). For comparison, we consider the influence of the parameters $\alpha \text{ and }c$ on the system performance in both cases with and without the batching policy. We define the traffic intensity in the w/o BATCH model $\rho_{w}$ by $\rho_{w}=\lambda/(c\mu)$
	
	Fig.~\ref{fig5} shows the effect of the setup rate on the performance measures for $\lambda=6$ and 18 (or the loads of $\rho_{w}=0.3$ and 0.9). We observe that the power consumption, mean response time, and ERP decrease with the setup rate and eventually stabilize at a fixed value. Notably, when the setup rate is extremely low, the power consumption curves in two cases of $\lambda=6$ and 18 overlap, regardless of whether the batching policy is used. This occurs because when the setup time is exceptionally long, all servers are in setup mode for most of the time. The batching policy reduces the power consumption when the setup rate is not too low and increases the power consumption when the setup time is large. However, the BATCH model consistently outperforms the w/o BATCH model regarding the mean response time and ERP. As a result, our batching policy enhances system performance, especially under high arrival rates.
	
	Fig.~\ref{fig6} illustrates the effect of the number of servers on the performance measures in both cases with and without the batching policy. We plot three curves for each case with $\alpha= 0.01, 0.1, 1$ while maintaining the load at $\rho_{w}=0.9$ (i.e., 90\% load). We observe that $Cost$ and ERP increase with the number of servers while the mean response time decreases with $c$ and converges to a fixed value as expected. The BATCH model consistently outperforms the w/o BATCH model regarding mean response time and ERP, but the latter saves more power than the BATCH model for $\alpha= 0.01$. Despite consuming more energy when the setup time is exceptionally long, our BATCH model is generally superior to the model without the batching policy.
	
	Fig.~\ref{fig7} shows the performance measures plotted against the setup rate for the cases $b=2,4,8$ and $b=1$, that is, without batching policy. $Cost$ increases with the batch size when the setup rate is low, whereas increasing the batch size reduces $Cost$ when the setup time is not too long. In addition, increasing the batch size may lead to an increase in the mean response time and ERP. Hence, a batching policy with an appropriate batch size can balance the trade-off between power consumption and mean response time when the setup time is varied.
	
	\begin{figure}
		\centering
		\begin{subfigure}{0.32\linewidth}
			\includegraphics[width=\linewidth]{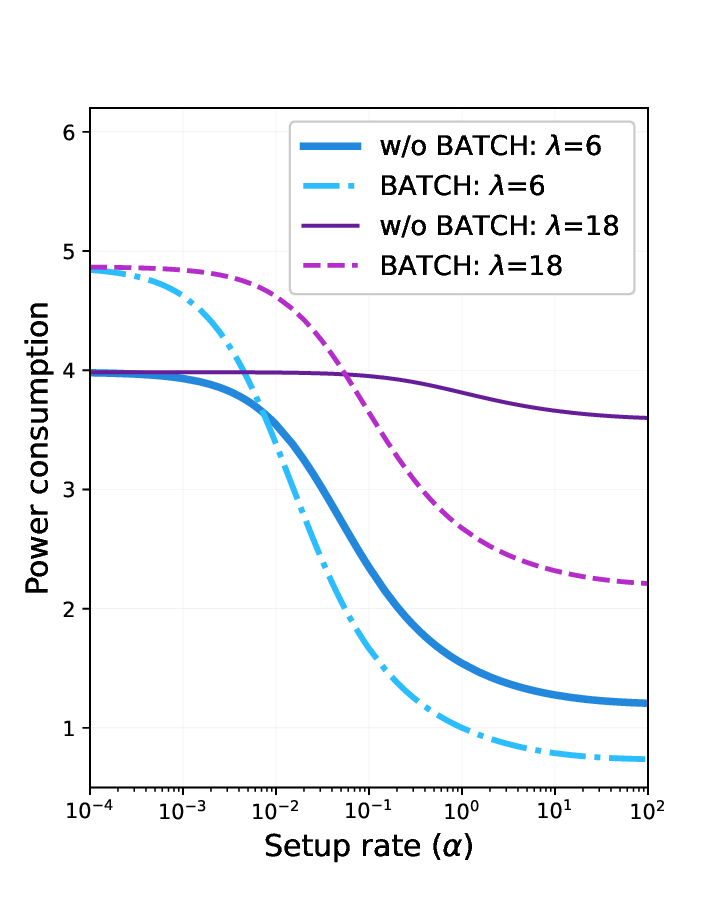}
			\caption{\textit{Cost} (W)}
			\label{fig5a}
		\end{subfigure}
		\begin{subfigure}{0.32\linewidth}
			\includegraphics[width=\linewidth]{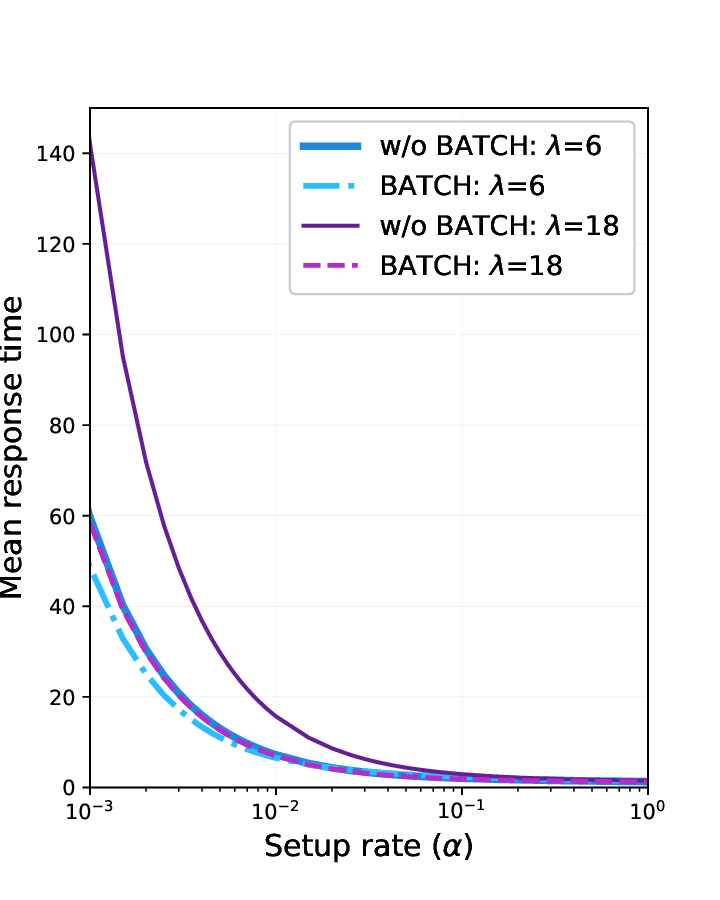}
			\caption{E$[W]$ (s)}
			\label{fig5b}
		\end{subfigure}
		\begin{subfigure}{0.32\linewidth}
			\includegraphics[width=\linewidth]{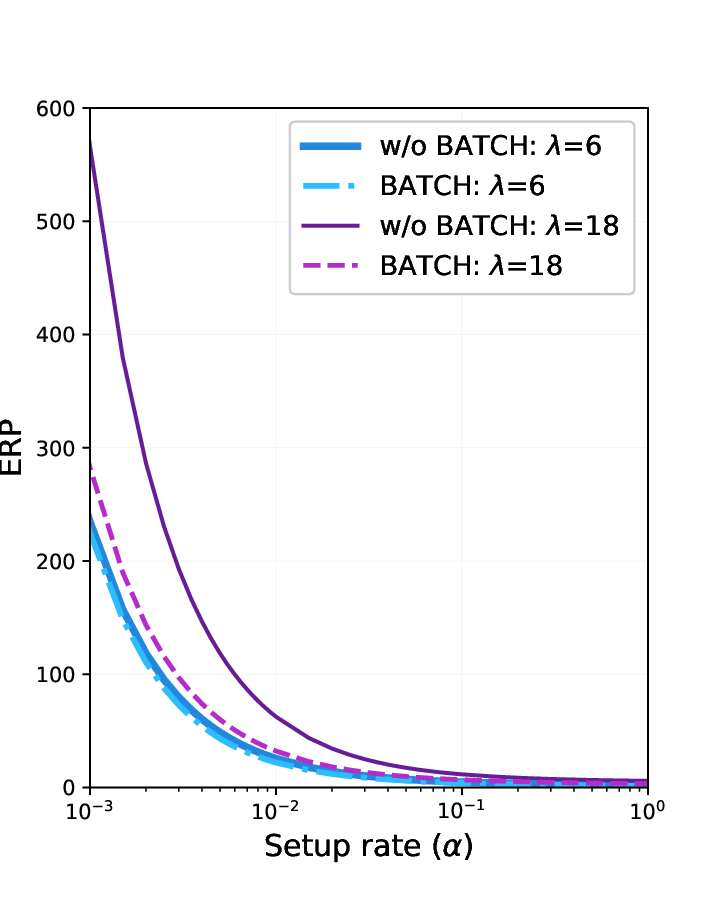}
			\caption{ERP (J)}
			\label{fig5c}
		\end{subfigure}
		\caption{ Comparison between the models without and with batching policy ($\lambda=6$ and 18).}
		\label{fig5}
	\end{figure}
	
	\begin{figure}
		\centering
		\begin{subfigure}{0.32\linewidth}
			\includegraphics[width=\linewidth]{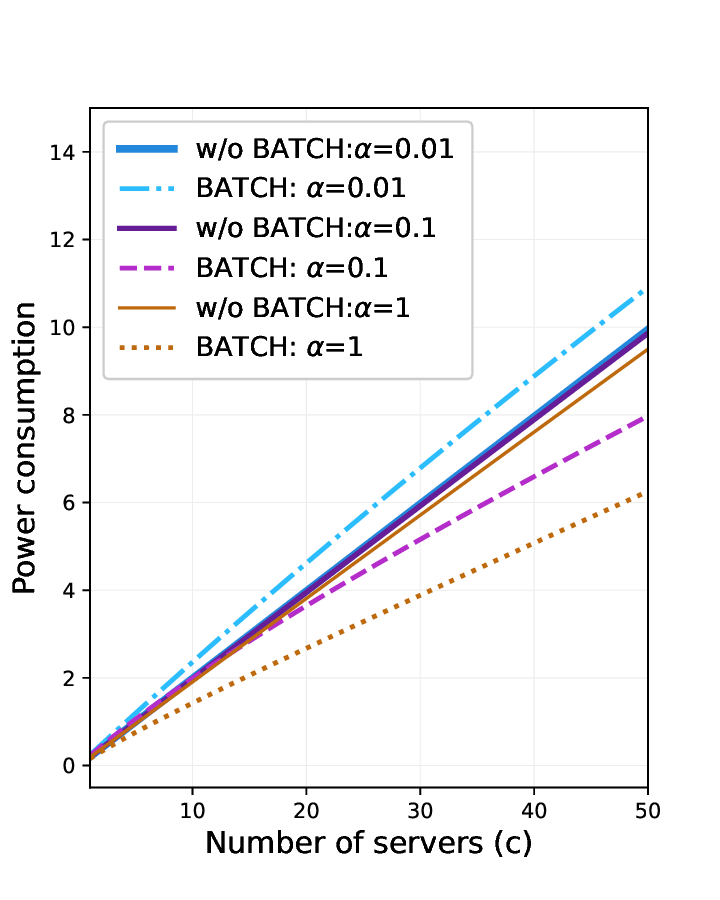}
			\caption{\textit{Cost} (W)}
			\label{fig6a}
		\end{subfigure}
		\centering
		\begin{subfigure}{0.32\linewidth}
			\includegraphics[width=\linewidth]{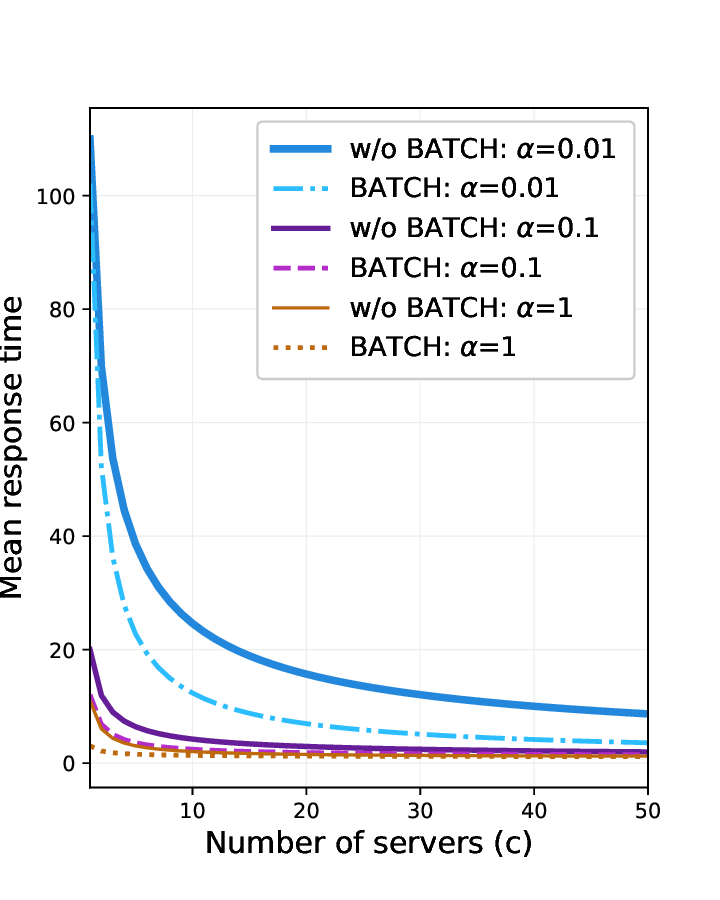}
			\caption{E$[W]$ (s)}
			\label{fig6b}
		\end{subfigure}
		\centering
		\begin{subfigure}{0.32\linewidth}
			\includegraphics[width=\linewidth]{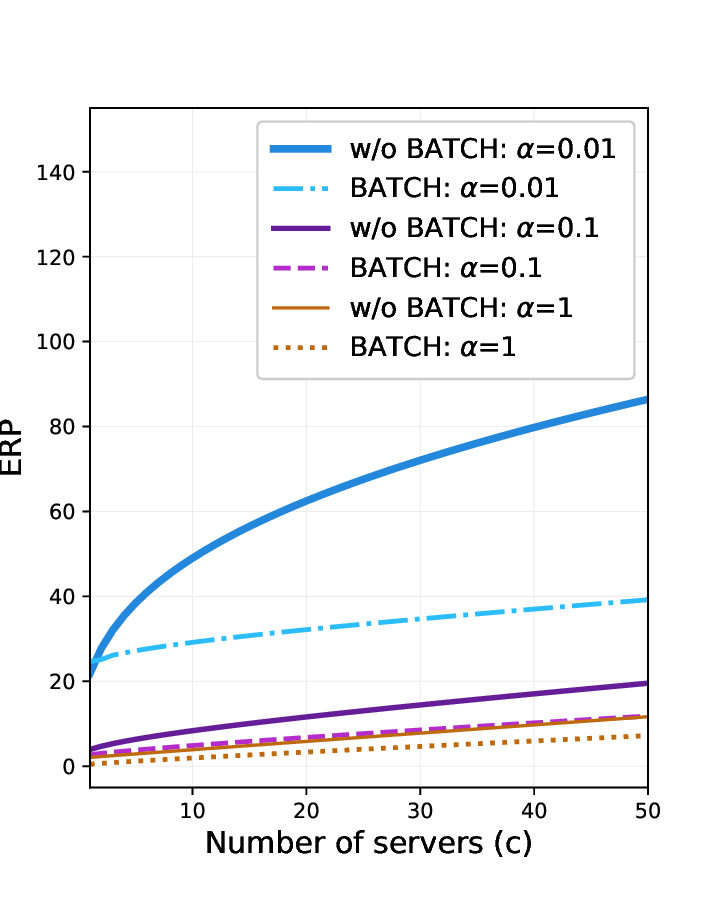}
			\caption{ERP (J)}
			\label{fig6c}
		\end{subfigure}
		\caption{ Comparison between the models without and with batching policy ($\alpha= 0.01, 0.1, 1$).}
		\label{fig6}
	\end{figure}
	
	\begin{figure}
		\centering
		\begin{subfigure}{0.32\linewidth}
			\includegraphics[width=\linewidth]{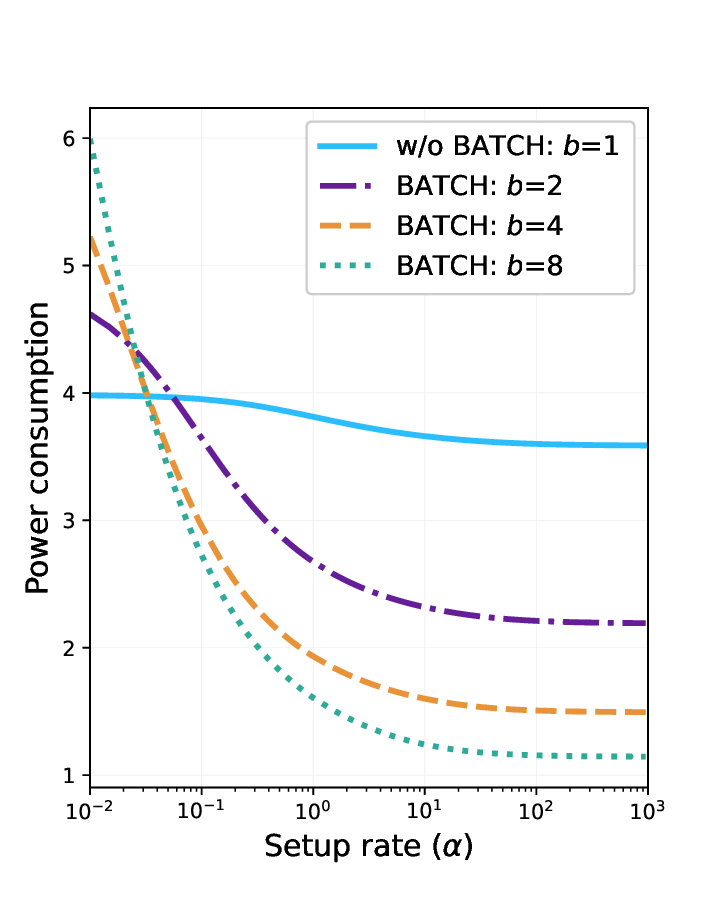}
			\caption{\textit{Cost} (W)}
			\label{fig7a}
		\end{subfigure}
		\begin{subfigure}{0.32\linewidth}
			\includegraphics[width=\linewidth]{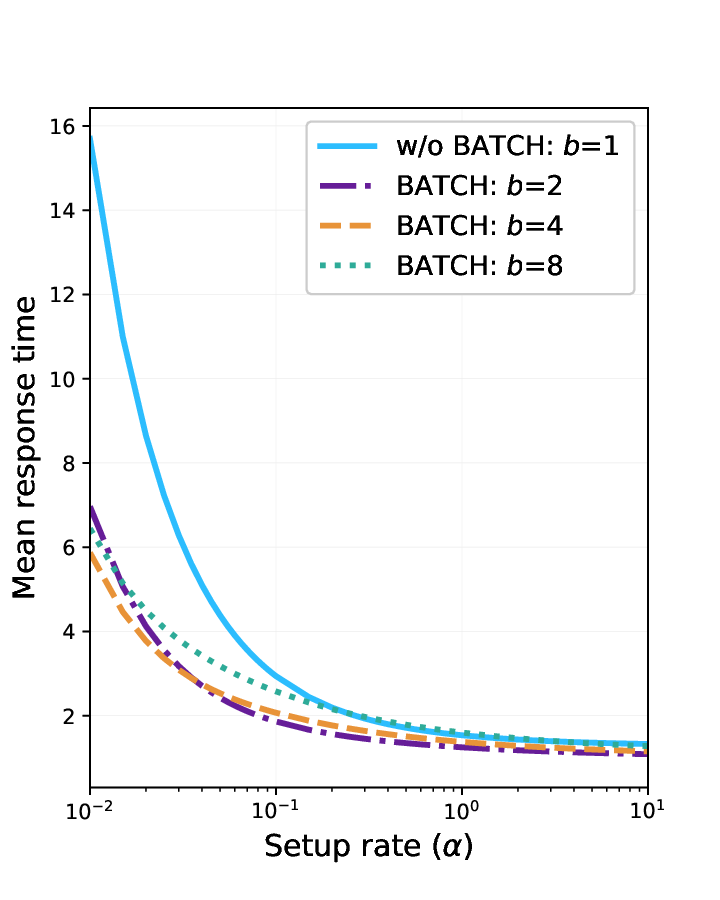}
			\caption{E$[W]$ (s)}
			\label{fig7b}
		\end{subfigure}
		\begin{subfigure}{0.32\linewidth}
			\includegraphics[width=\linewidth]{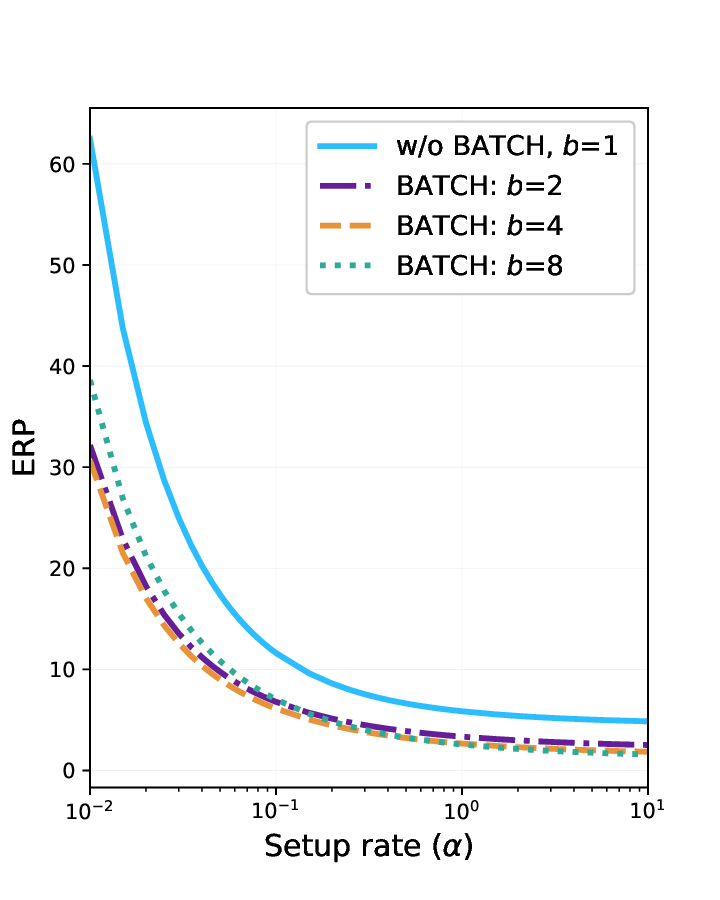}
			\caption{ERP (J)}
			\label{fig7c}
		\end{subfigure}
		\caption{ The effect of batch size on system performance ($\lambda=18$).}
		\label{fig7}
	\end{figure}
	
	\subsection{Some variant models}
	As discussed in Section \ref{sec6.1}, the M/M/$c$/SET--BATCH model generally outperforms the model without the batching policy. Therefore, in this section, we analyze the BATCH model further and compare it with its variants to find the optimal policies for enhancing system performance. In all the numerical examples, we define the traffic intensity by $\rho = \lambda/(cb\mu)$ for the M/M/$c$/SET--BATCH and  M/M/$c$/SET--BATCH/Delayoff models, and by $\rho = \lambda/(cb_2\mu)$ for the M/M/$c$/SET--BATCH/Dynamic model. We fix the mean service time $1/\mu = 1\text{s}$, the traffic intensity  $\rho=0.7$, and the number of servers $c=20$. 
	
	The performance of the delayoff policy depends on the choice of the delayoff rate $\beta$, whereas the performance of the dynamic batch-size policy depends on the batch-size range ($b_1, b_2$). Therefore, before comparing the above policies, we present how we choose the parameter $\beta$ for the M/M/$c$/SET--BATCH/Delayoff model and the batch-size range ($b_1, b_2$) for the M/M/$c$/SET--BATCH/Dynamic model.
	\subsubsection{M/M/$c$/SET--BATCH/Delayoff}
	We fix the batch size $b=5$ and the traffic intensity $\rho=0.7$, and plot three curves with $\alpha= 0.1, 1, 10$. Fig.~\ref{fig8} shows the power consumption, mean response time, and ERP plotted against the delayoff rate $\beta$ for the M/M/$c$/SET--BATCH/Delayoff model. The mean response time and ERP increase when the delay-off rate increases, as expected, due to the increase in setup time. Moreover, for delayoff rates under a certain threshold, the power consumption decreases as $\beta$ increases since we reduce the power consumption in idle mode. However, beyond this threshold, the power consumption decreases as $\beta$ decreases because we reduce the power consumption in setup mode by keeping servers idle longer. Regarding power consumption, the optimal value of $\beta$ increases with the setup rate $\alpha$, while $\beta=0$ is optimal for minimizing both the mean response time and the ERP in heavy traffic. Maccio and Down \cite{maccio2015optimal} also showed that the optimal delayoff rate is $\beta=0$ in heavy traffic and $\beta=\infty$ in light traffic for the ERWS cost metric (weighted sum of delay and power) in a similar delayoff model without batch services. In addition, Gebrehiwot et al. \cite{gebrehiwot2016optimal} proved that the optimal delayoff rate is either $\beta=0$ or $\beta=\infty$ for both the ERWS and ERP cost metrics in a single-server system even when service and setup times follow general distributions.
	\begin{figure}
		\centering
		\begin{subfigure}{0.32\linewidth}
			\includegraphics[width=\linewidth]{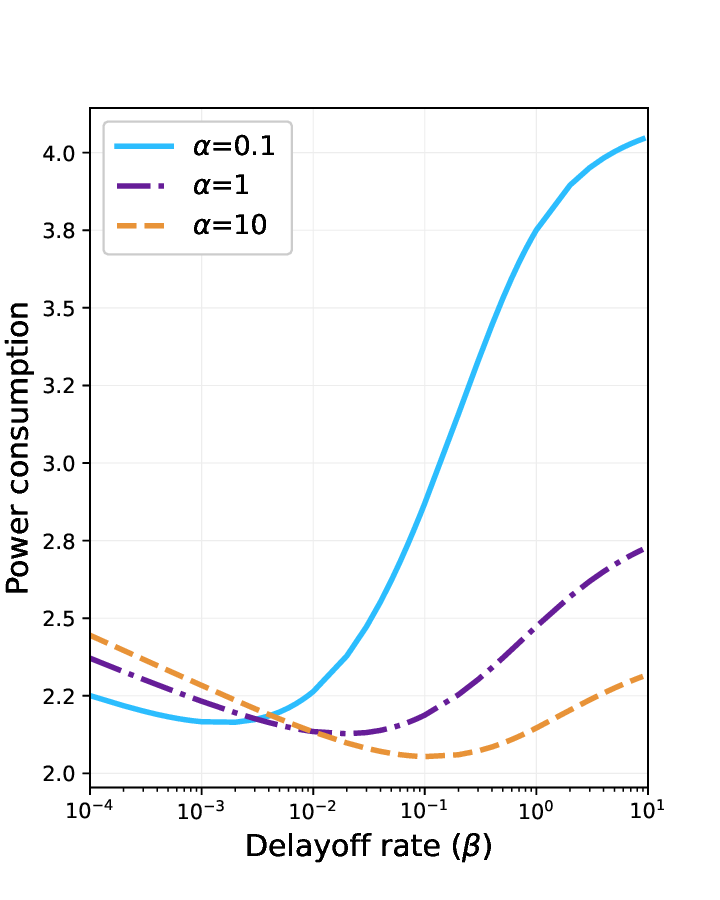}
			\caption{\textit{Cost} (W)}
			\label{fig8a}
		\end{subfigure}
		\centering
		\begin{subfigure}{0.32\linewidth}
			\includegraphics[width=\linewidth]{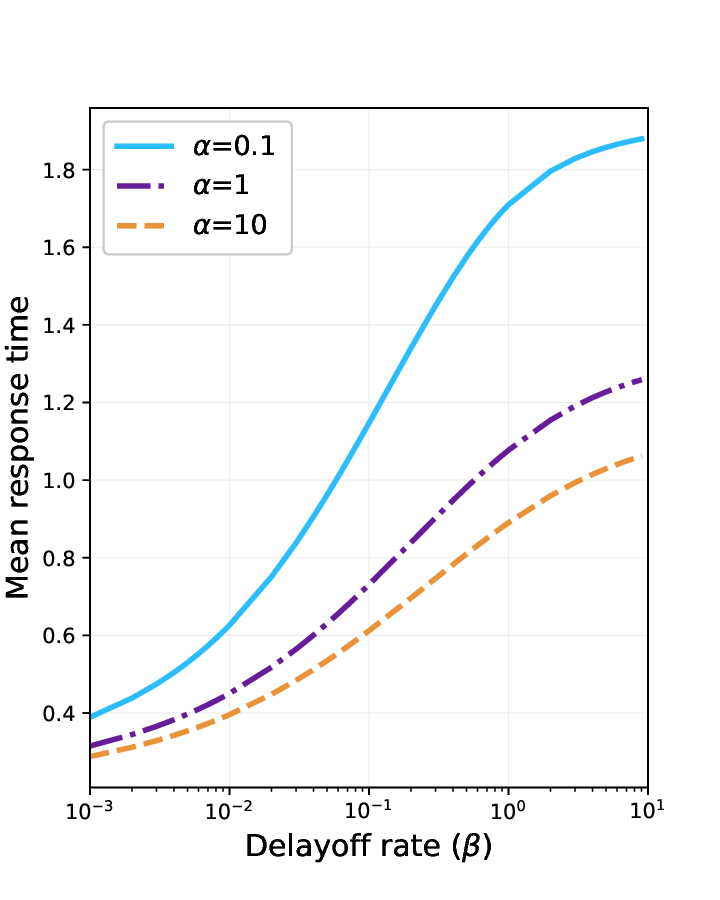}
			\caption{E$[W]$ (s)}
			\label{fig8b}
		\end{subfigure}
		\begin{subfigure}{0.32\linewidth}
			\includegraphics[width=\linewidth]{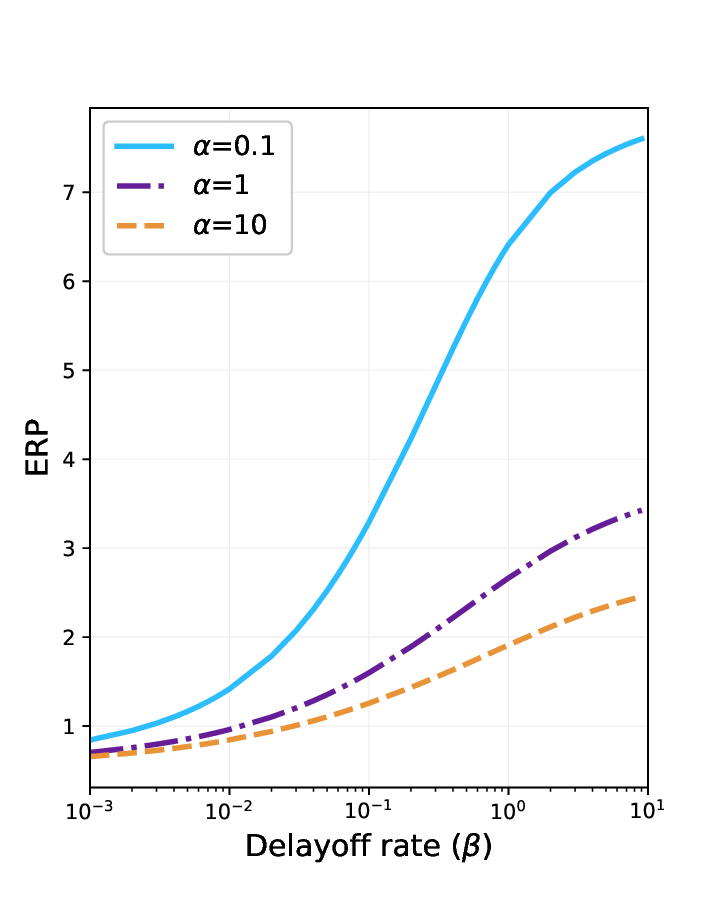}
			\caption{ERP (J)}
			\label{fig8c}
		\end{subfigure}
		\caption{System performance versus delayoff rate $\beta$.}
		\label{fig8}
	\end{figure}
	\subsubsection{M/M/$c$/SET--BATCH/Dynamic}\label{sec6.2}
	To determine the optimal batch size range ($b_1, b_2$), we fix the traffic intensity $\rho=0.7$ and the upper limit $b_2=5$, and vary the lower limit $b_1=1,2,3,4$. Fig.~\ref{fig9} illustrates the power consumption, mean response time, and ERP for the M/M/$c$/SET--BATCH/Dynamic model. For values of $\alpha$ under a threshold, the power consumption, mean response time, and ERP increase as $b_1$ increases. Thus, when the setup rate is small, it is better to set up a server whenever a job arrives. However, beyond the threshold of $\alpha$, the power consumption, E[$W$], and ERP decrease as $b_1$ increases. Therefore, the M/M/$c$/SET--BATCH model outperforms the M/M/$c$/SET--BATCH/Dynamic model when the setup rate is high.
	\begin{figure}
		\centering
		\begin{subfigure}{0.32\linewidth}
			\includegraphics[width=\linewidth]{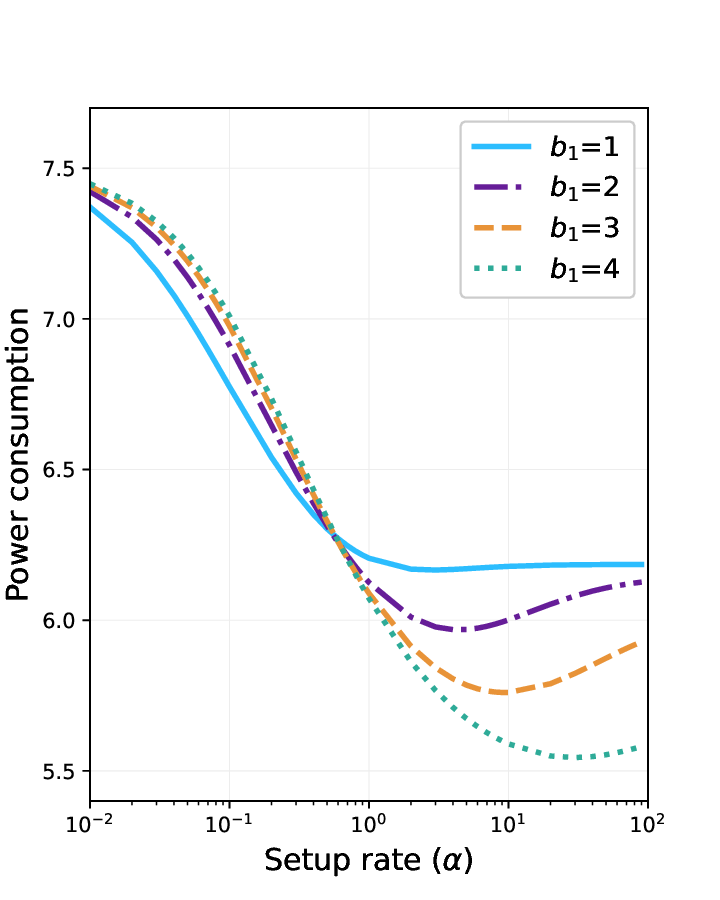}
			\caption{\textit{Cost} (W)}
			\label{fig9a}
		\end{subfigure}
		\centering
		\begin{subfigure}{0.32\linewidth}
			\includegraphics[width=\linewidth]{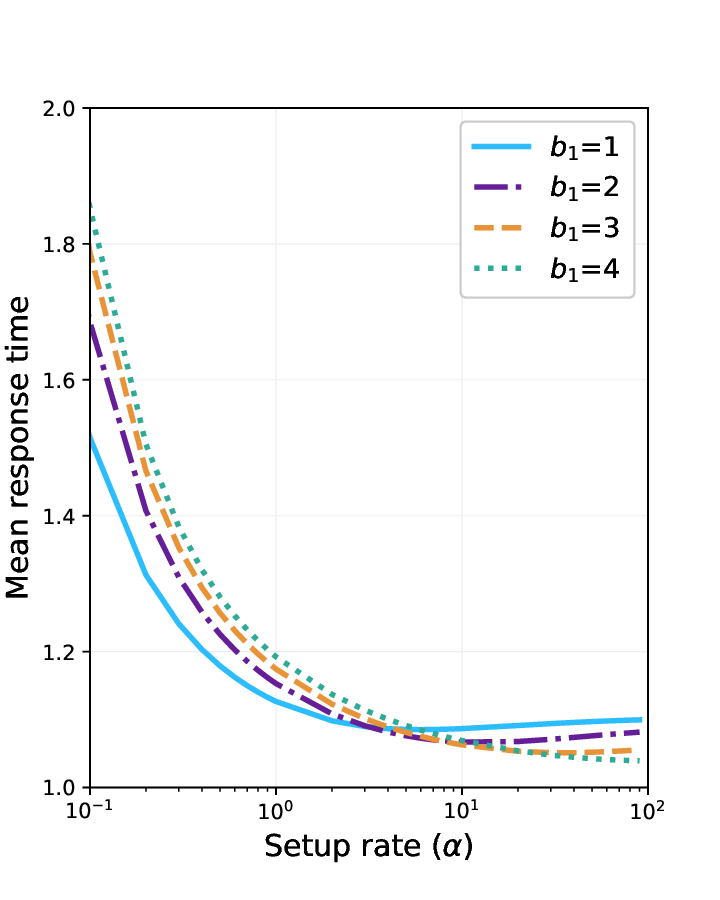}
			\caption{E$[W]$ (s)}
			\label{fig9b}
		\end{subfigure}
		\begin{subfigure}{0.32\linewidth}
			\includegraphics[width=\linewidth]{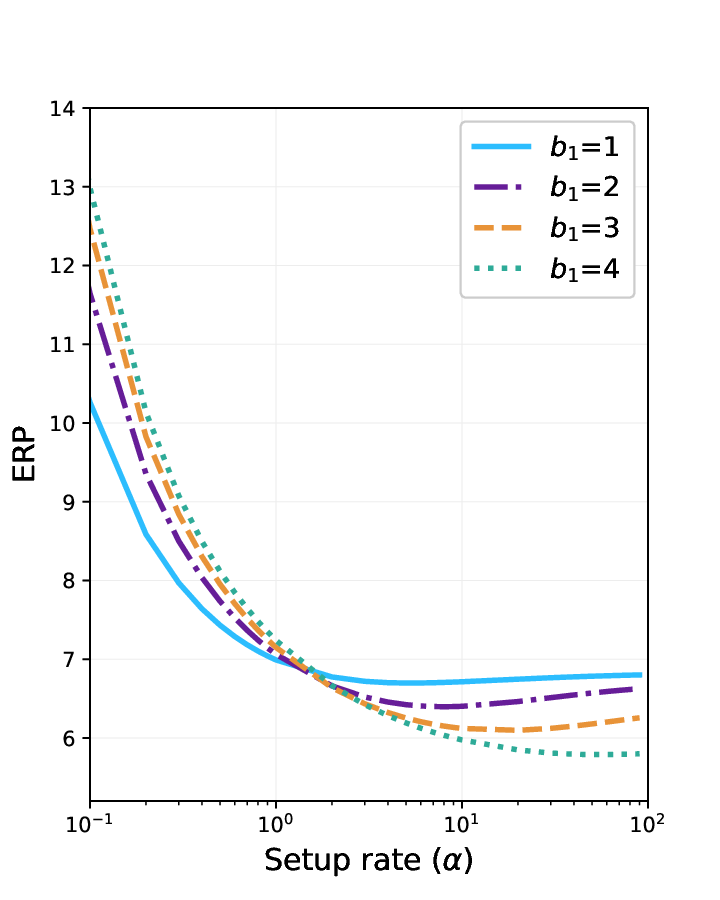}
			\caption{ERP (J)}
			\label{fig9c}
		\end{subfigure}
		\caption{System performance for various values of $b_1$.}
		\label{fig9}
	\end{figure}
	\subsubsection{Comparison of all policies}
	\begin{figure}
		\centering
		\begin{subfigure}{0.32\linewidth}
			\includegraphics[width=\linewidth]{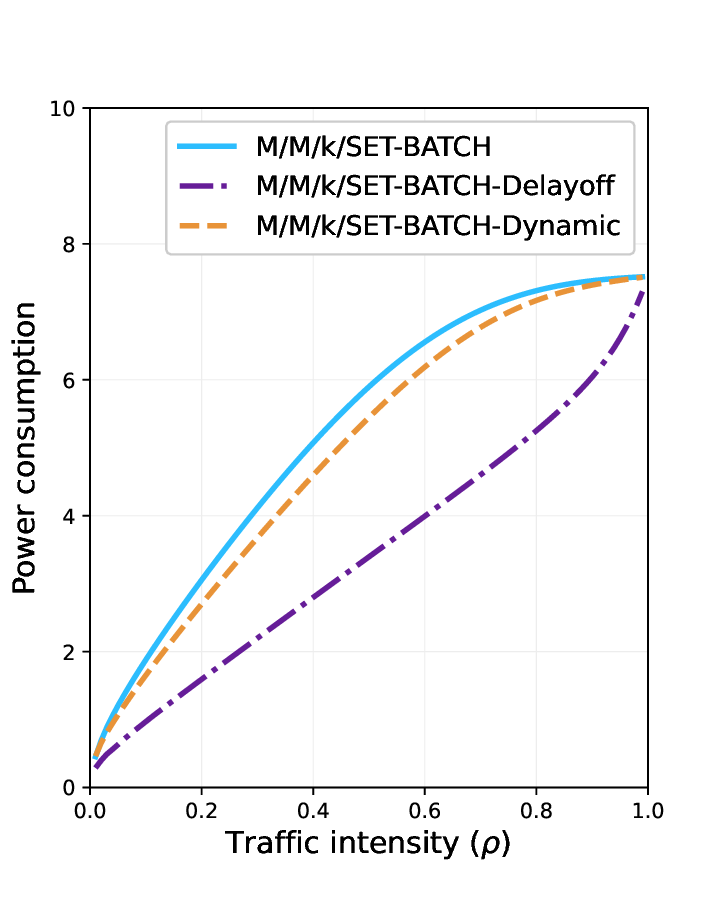}
			\caption{\textit{Cost} (W)}
			\label{fig10a}
		\end{subfigure}
		\centering
		\begin{subfigure}{0.32\linewidth}
			\includegraphics[width=\linewidth]{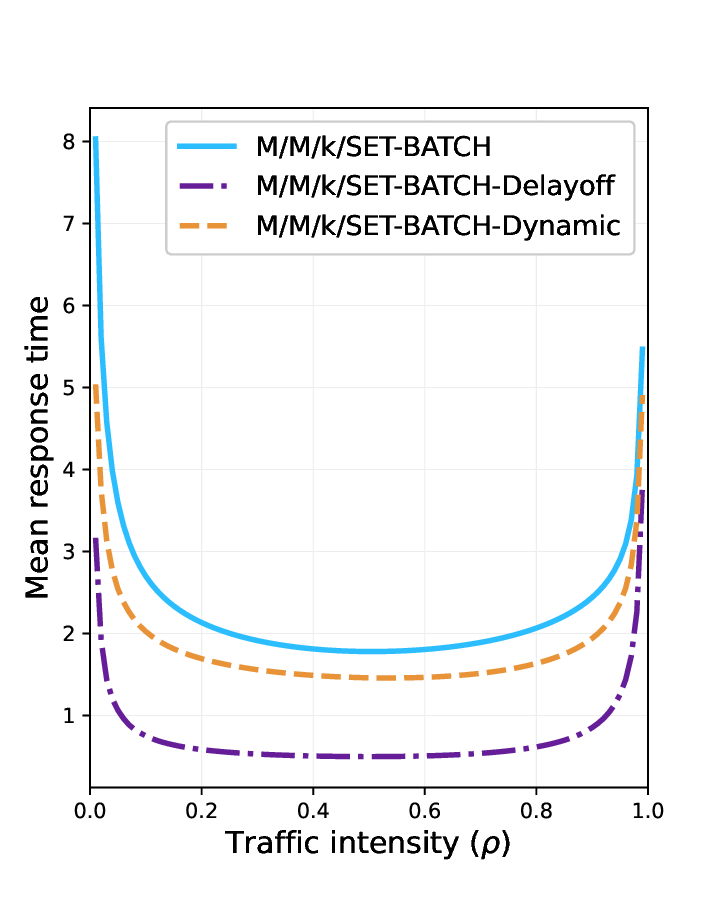}
			\caption{E$[W]$ (s)}
			\label{fig10b}
		\end{subfigure}
		\begin{subfigure}{0.32\linewidth}
			\includegraphics[width=\linewidth]{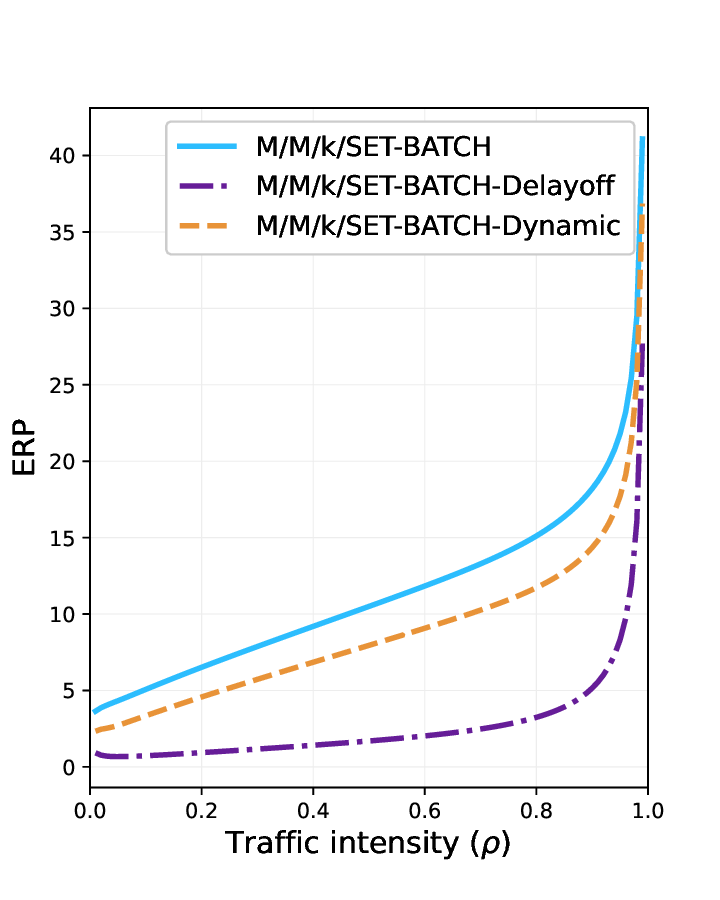}
			\caption{ERP (J)}
			\label{fig10c}
		\end{subfigure}
		\caption{System performance versus $\rho$ for all policies ($1/\alpha= 10 \text{s}$).}
		\label{fig10}
	\end{figure}
	\begin{figure}
		\centering
		\begin{subfigure}{0.32\linewidth}
			\includegraphics[width=\linewidth]{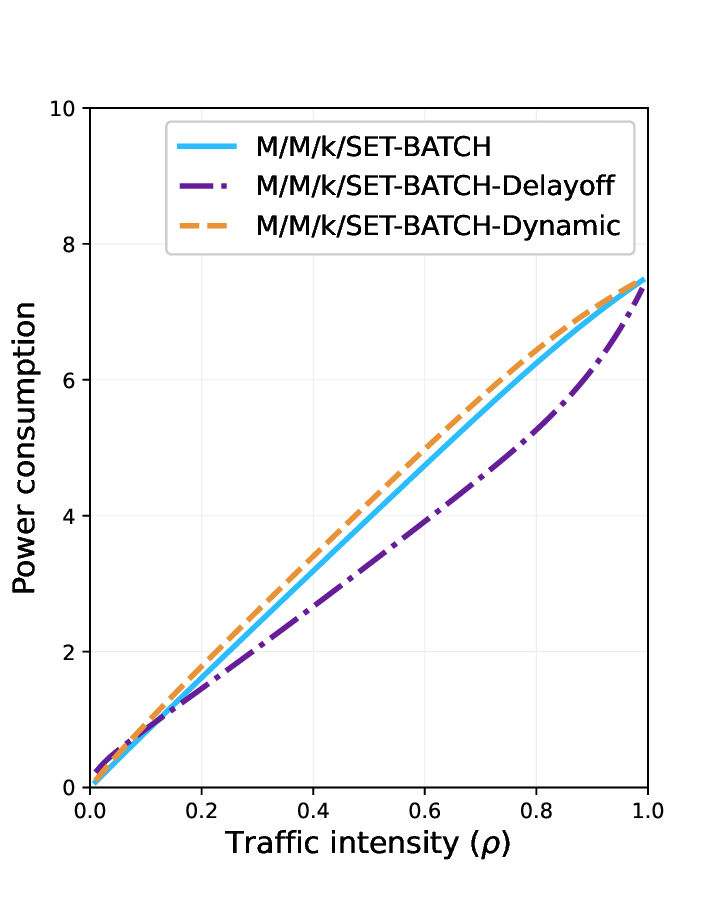}
			\caption{\textit{Cost} (W)}
			\label{fig11a}
		\end{subfigure}
		\centering
		\begin{subfigure}{0.32\linewidth}
			\includegraphics[width=\linewidth]{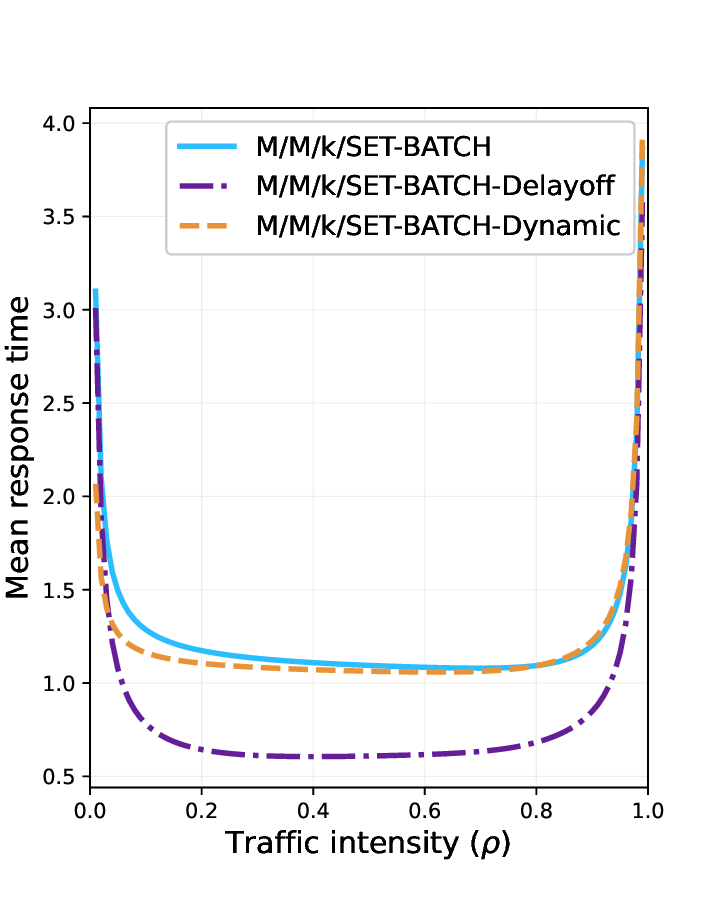}
			\caption{E$[W]$ (s)}
			\label{fig11b}
		\end{subfigure}
		\begin{subfigure}{0.32\linewidth}
			\includegraphics[width=\linewidth]{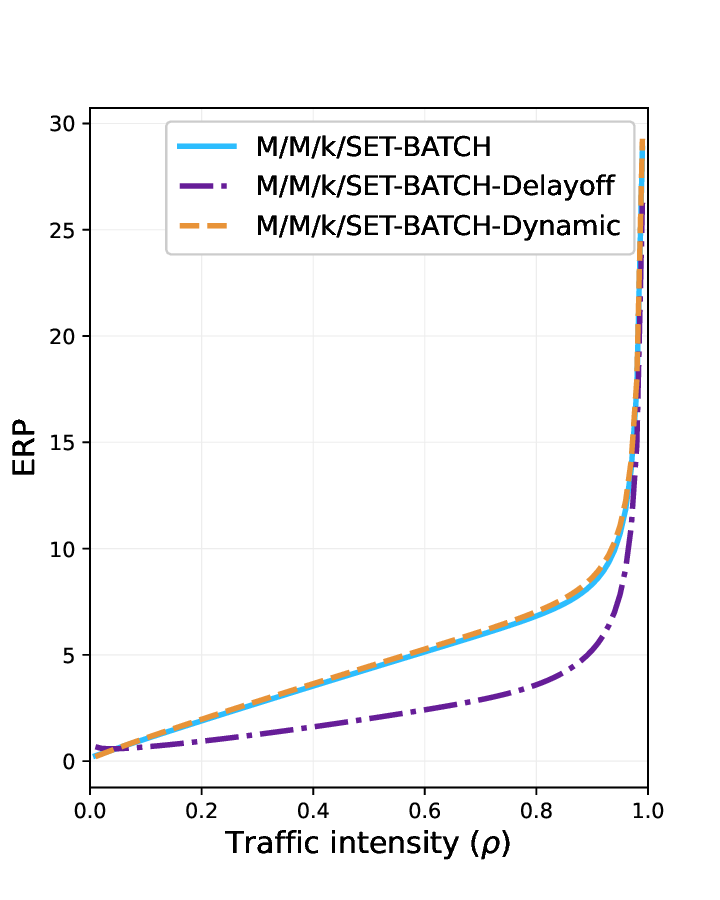}
			\caption{ERP (J)}
			\label{fig11c}
		\end{subfigure}
		\caption{System performance versus $\rho$ for all policies ($1/\alpha= 0.1 \text{s}$).}
		\label{fig11}
	\end{figure}
	
	Fig.~\ref{fig10} shows the power consumption, the mean response time, and ERP plotted against the traffic intensity $\rho$ with a mean setup time of $1/\alpha= 10 \text{s}$. In this case, we chose a mean delayoff time $1/\beta= 200 \text{s}$ and batch sizes $b=5$, $b_1= 1, b_2=5$. Comparing the M/M/$c$/SET--BATCH and M/M/$c$/SET--BATCH/Dynamic models, we observe that the M/M/$c$/SET--BATCH model has a much higher $Cost$, E[$W$], and ERP. Thus, using the dynamic batch size policy enhances the system performance of the M/M/$c$/SET--BATCH model. In addition, we observe that the M/M/$c$/SET--BATCH/Delayoff model has the lowest $Cost$, E[$W$], and ERP among all models. Thus, the M/M/$c$/SET--BATCH/Delayoff model outperforms all the other models for a small setup rate.
	
	Fig.~\ref{fig11} shows the power consumption, the mean response time, and ERP plotted against the traffic intensity $\rho$ for a mean setup time $1/\alpha= 0.1$s. In this case, we assume a mean delayoff time $1/\beta= 10$s, and batch sizes $b=5$, $b_1=3, b_2=5$. Although the M/M/$c$/SET--BATCH/Dynamic model has a lower E[$W$] than the M/M/$c$/SET--BATCH model under light traffic, it has a higher $Cost$ than the other models. Thus, the dynamic batch size is not a good policy for high setup rates, especially under high traffic. Similar to the case of $1/\alpha= 10$s, the M/M/$c/$SET--BATCH/Delayoff model is generally superior to all the other models except for the case of extremely low traffic intensity. Indeed, under a low traffic regime, delaying the shutdown of a server after service completion results in higher power consumption than immediately turning it off and restarting it when needed, especially when the setup time is short. Hence, the optimal delayoff rate to minimize the ERP under light traffic regime is $\beta = \infty$, which is consistent with the results in \cite{maccio2015optimal} and \cite{gebrehiwot2016optimal}. In addition, the longer the setup time, the higher the percentage change in system performance under different policies.
	
	\section{Conclusions}\label{sec7}
	We have presented an analysis of the M/M/$c$/Setup model with the batching policy and some of its variations, including the M/M/$c$/SET--BATCH/ Delayoff and the M/M/$c$/SET--BATCH/Dynamic models. We use a matrix analytic approach to derive an exact recursive algorithm for the stationary distribution and performance measures of the M/M/$c$/SET--BATCH model, whereas a matrix geometric method is used for the other models. Notably, in the M/M/$c$/SET--BATCH model, the special structure of the rate matrices allows us to obtain them with a lower computational complexity. In addition, we used simulations to analyze some variations of the M/M/$c$/SET--BATCH model involving various setup time distributions: hyperexponential, Erlang, and deterministic setup time models. The simulations show that variable setup times give better results for the setup models with the batching policy. Comparing the M/M/$c$/SET--BATCH and M/M/$c$/Setup models, we see that the batching policy enhances system performance, especially under high arrival rates and prolonged setup times. Finally, the dynamic batch size policy improves the system performance of the M/M/$c$/SET--BATCH model when the setup time is long, and the delayoff policy generally outperforms all the other policies except for the case of short setup time and extremely low traffic intensity. In such cases, maintaining idle servers after service completion is unnecessary, making the M/M/$c$/SET--BATCH model a more suitable choice. Although we have only considered the policies independently, combining them to enhance the system performance is possible.
	
	\section*{CRediT authorship contribution statement}
	\textbf{Thu Le-Anh:} Writing -- review \& editing, Writing -- original draft, Visualization, Validation, Software, Methodology, Investigation, Formal analysis, Data curation, Conceptualization. \textbf{Tuan Phung-Duc:} Writing -- review \& editing, Validation, Supervision, Resources, Project administration, Methodology, Investigation, Funding acquisition, Formal analysis, Conceptualization.
	
	\section*{Declaration of competing interest}
	We do not have conflicts to declare.
	
	\section*{Acknowledgments}
	The research of Tuan Phung-Duc was supported in part by JSPS KAKENHI, Japan, Grant Number 21K11765 and by F-MIRAI: R and D Center for Frontiers of MIRAI in Policy and Technology, Japan, the University of Tsukuba and Toyota Motor Corporation collaborative R and D center, Japan.
	
	\appendix
	\section{Proofs of Section \ref{sec4}}
	
	\subsection{Proof of Lemma \ref{lem:1}} \label{Ad:1}
	\begin{proof} 
		We now state Rouché’s theorem \cite{adan2006application}, which we will need to prove Lemma \ref{lem:1}.
		
		\begin{theorem}[Rouché’s Theorem]
			Let the bounded region $D$ have a simple closed contour C as its boundary. Let $f(z)$ and $g(z)$ be analytic both in $D$ and on $C$. Assume that $\lvert f(z) \rvert < \lvert g(z) \rvert$ on $C$. Then $f(z) - g(z)$ and $g(z)$ have the same number of zeros in $D$; all zeros are counted according to their multiplicity.
		\end{theorem}
		Let $ f(z)=\lambda+i\mu z^{b+1}$, $g(z)= (\lambda+ i\mu + (c-i)\alpha) z$, and $h(z)= f(z) - g(z)$. Accordingly, $f(z)$ and $g(z)$ are analytic on the unit disk, $ \lvert z \rvert \leq 1$. In addition, $\lvert f(z) \rvert <\lvert g(z) \rvert$ on the unit circle, $ \lvert z \rvert = 1$. 
		
		Thus $h(z)$ has one zero, as does $g(z)$, on the unit disk. As a result, Eq.~(\ref{eq:2}) has a unique root in (0,1).
	\end{proof}
	
	\subsection{Proof of Lemma \ref{lem:2}} \label{Ad:2}
	\begin{proof}
		We use induction to prove that the recursive sequence $\{x_n\}$ is monotonically increasing, i.e.,
		\begin{align*}
			x_{n+1}> x_{n }, \quad \text{for} \quad n\in \mathbb{N}. 
		\end{align*}
		For the case $n=0$, we have 
		\begin{align*}
			x_{1} = p \geq x_{0}.
		\end{align*}
		Thus, the statement holds in the case $n=0$. Assume that the statement holds for some positive integer $k$. Then,
		\begin{align*}
			x_{k+2}-x_{k+1}= q( x_{k+1}^{b+1}-x_{k}^{b+1}) > 0,
		\end{align*}
		since $x_{k+1} > x_k$ by the induction hypothesis. Hence, $x_{n+1} > x_n$ for all positive integers $n$. 
		
		Next, we use induction again to prove that the sequence is bounded above by 1. The hypothesis holds for $n = 0$ since $x_1 = p < 1$. Assume that $x_k < 1$ for all positive integers up to $k$. Then,
		\begin{align*}
			x_{k+1} = p+ q x_{k}^{b+1} < p+q < 1.
		\end{align*}
		Hence, $x_n < 1$ for all positive integers $n$.
	\end{proof}
	
	\section*{Data availability}
	Data will be made available on request.

\end{document}